\RequirePackage{fix-cm}
\documentclass[smallcondensed]{svjour3}     
\usepackage{epsfig,epic,eepic,units}
\usepackage{url}
\usepackage{longtable}
\usepackage{mathrsfs}
\usepackage{multirow}
\usepackage{bigstrut}
\usepackage{amssymb}
\usepackage{graphicx}
\usepackage{amsmath}

\usepackage[ruled, vlined, linesnumbered]{algorithm2e}

\SetKwIF{If}{ElseIf}{Else}{$\mathsf{if}$}{$\mathsf{then}$}{}{$\mathsf{else}$}{}

\SetKwBlock{Begin}{$\mathsf{begin}$}{$\mathsf{end}$}

\usepackage{latexsym}

\usepackage{amsmath}

\usepackage{amssymb}

\usepackage{wrapfig} 

\usepackage{color}
\usepackage{booktabs}
\usepackage{graphicx}
\usepackage{multirow}
\usepackage{url}
\usepackage{wrapfig}
\smartqed  
\usepackage{graphicx}
\begin{document}

\newcommand{\mike}[1]{\textcolor{red}{\sc Mike: #1}}
 
\title{Reasoning about Algebraic Data Types with Abstractions
}


\author{Tuan-Hung Pham         \and Andrew Gacek \and
        Michael W. Whalen
}


\institute{Tuan-Hung Pham \at
              Department of Computer Science and Engineering, University of Minnesota \\
              \email{hung@cs.umn.edu}           
           \and
           Andrew Gacek \at
              Rockwell Collins, Advanced Technology Center \\
              \email{andrew.gacek@gmail.com}
           \and
           Michael W. Whalen \at
              Department of Computer Science and Engineering, University of Minnesota \\
              \email{whalen@cs.umn.edu}           
}

\date{Received: date / Accepted: date}

\maketitle
\urldef{\toolurl}\url{http://crisys.cs.umn.edu/rada}

\newcommand{\smtlib}[0]{SMT-LIB 2.0}
\newcommand{\kw}[1]{\textsf{#1}}
\newcommand{\Leaf}[0]{\kw{Leaf}}
\newcommand{\Node}[0]{\kw{Node}}
\newcommand{\leaf}[0]{\kw{leaf}}
\newcommand{\node}[0]{\kw{node}}
\newcommand{\SLeaf}[0]{\kw{SLeaf}}
\newcommand{\SNode}[0]{\kw{SNode}}
\newcommand{\List}[0]{\kw{List}}
\newcommand{\Some}[0]{\kw{Some}}
\newcommand{\None}[0]{\kw{None}}
\newcommand{\true}[0]{\kw{true}}
\newcommand{\false}[0]{\kw{false}}
\newcommand{\ite}[0]{\kw{ite}}
\newcommand{\combine}[0]{\kw{combine}}
\newcommand{\uf}[0]{U_\alpha}
\newcommand{\unrolledtrees}[0]{FN}

\newcommand{\constructor}{\mathbb{C}}
\newcommand{\selector}{\mathbb{S}}
\newcommand{\tester}[1]{#1?}
\newcommand{\rdttype}{\tau}
\newcommand{\argtype}{s}
\newcommand{\bigargtype}{S}
\newcommand{\argval}{x}
\newcommand{\arity}{n}

\newcommand{\solver}{\mathcal{S}}
\newcommand{\theory}{\mathcal{T}}
\newcommand{\collection}{\mathcal{C}}
\newcommand{\logic}{\mathcal{L}}
\newcommand{\func}{\mathcal{F}}
\newcommand{\elem}{\mathcal{E}}
\newcommand{\norm}[1]{\mbox{#1}}
\newcommand{\mybrack}[1]{[ #1 ]}

\newcommand{\sig}{sig}  


\newcommand{\subtree}{\preceq}
\newcommand{\strictsubtree}{\precneqq}
\newcommand{\strictsupertree}{\succneqq}
\newcommand{\naturalnums}{\mathbb{N}}
\newcommand{\oddnums}{\bar\naturalnums}
\newcommand{\halpha}{h_\alpha}
\newcommand{\minbeta}{\mathcal{M_\beta}}
\newcommand{\treessameheight}{\mathcal{T}_{height}}
\newcommand{\catalan}{\mathbb{C}}
\newcommand{\etoc}{\delta}

\newcommand{\operator}{\oplus}
\newcommand{\identity}[1]{id_#1}

\newcommand{\dryad}{\sc{Dryad}}
\newcommand{\strand}{\sc{Strand}}
\newcommand{\model}{\mathcal{M}}

\newcommand{\unrollingdepth}{\mathfrak{D}}

\newcommand{\predicate}{\kw{pr}}
\newcommand{\cpredicate}{c_\predicate}

\newcommand{\userpredicate}{pr_u}

\newcommand{\cleaf}{c_\leaf}

\newcommand{\isrecursive}{\mathsf{rec}}
\newcommand{\alphaac}{\alpha_{\ensuremath{AC}}}
\newcommand{\alphapac}{\alpha_{\ensuremath{PAC}}}

\newcommand*{\exampleEndMark}{\hfill\ensuremath{\vartriangle}}

\newcommand{\height}{\ensuremath{\mathit{height}}}
\newcommand{\size}{\ensuremath{\mathit{size}}}
\newcommand{\shape}{\ensuremath{\mathit{shape}}}
\newcommand{\range}{\ensuremath{\mathit{range}}}
\newcommand{\DW}{\ensuremath{\mathit{DW}}}
\newcommand{\remDirtyWords}{\ensuremath{\mathit{remDirtyWords}}}
\newcommand{\dirty}{\ensuremath{\mathit{dirty}}}
\newcommand{\setOf}{\ensuremath{\mathit{setOf}}}
\newcommand{\UF}{\ensuremath{\mathit{UF}}}
\newcommand{\unrollStep}{\ensuremath{\mathit{unrollStep}}}
\newcommand{\Mirror}{\ensuremath{\mathit{Mirror}}}
\newcommand{\numshapes}{\ensuremath{\mathit{ns}}}

\begin{abstract}


Reasoning about functions that operate over algebraic data types is an
important problem for a large variety of applications. One application
of particular interest is network applications that manipulate or
reason about complex message structures, such as XML messages. This
paper presents a decision procedure for reasoning about algebraic data
types using abstractions that are provided by catamorphisms: fold
functions that map instances of algebraic data types to values in a
decidable domain. We show that the procedure is sound and complete for
a class of catamorphisms that satisfy a generalized sufficient
surjectivity condition. Our work extends a previous decision procedure
that unrolls catamorphism functions until a solution is found.

We use the generalized sufficient surjectivity condition to address an
incompleteness in the previous unrolling algorithm (and associated
proof). We then propose the categories of monotonic and associative
catamorphisms, which we argue provide a more intuitive inclusion test
than the generalized sufficient surjectivity condition. We use these
notions to address two open problems from previous work: (1) we
provide a bound, with respect to formula size, on the number of
unrollings necessary for completeness, showing that it is linear for
monotonic catamorphisms and exponentially small for associative
catamorphisms, and (2) we demonstrate that associative catamorphisms
can be combined within a formula while preserving completeness. Our
combination results extend the set of problems that can be reasoned
about using the catamorphism-based approach.

We also describe an implementation of the approach, called RADA, which
accepts formulas in an extended version of the \smtlib\ syntax. The
procedure is quite general and is central to the reasoning
infrastructure for Guardol, a domain-specific language for reasoning
about network guards.

\end{abstract}

\section{Introduction}
\label{section:introduction}
Decision procedures have been a fertile area of research in recent years,
with several advances in the breadth of theories that can be decided and the speed with which substantial problems can be solved.
When coupled with SMT solvers, these procedures can be combined and used to solve complex formulas relevant to software and hardware verification.
An important stream of research has focused on decision procedures for algebraic data types.
Algebraic data types are important for a wide variety of problems: they provide a natural representation for tree-like structures such as abstract syntax trees and XML documents;
they are also the fundamental representation of recursive data for functional programming languages.

Algebraic data types provide a significant challenge for decision procedures since they are recursive and usually unbounded in size.
Early approaches focused on equalities and disequalities over the structure of elements of data types~\cite{BarSTPDPAR06,Oppen:1980:RRD}.
While important, these structural properties are often not expressive enough to describe interesting properties involving the data stored in the data type.
Instead, we often are interested in making statements both about the structure and contents of data within a data type.
For example, one might want to express that all integers stored within a tree are positive or that the set of elements in a list does not contain a particular value.

Suter et al. described a parametric decision procedure for reasoning about algebraic data types using catamorphism (fold) functions \cite{Suter2010DPA}.
In the procedure, catamorphisms describe the abstract views of the data type that can then be reasoned about in formulas.
For example, suppose that we have a binary tree data type with functions to add and remove elements from the tree, as well as check whether an element was stored in the tree.
Given a catamorphism {\em setOf} that computes the set of elements stored in the tree, we could describe a specification for an {\em add} function as:
\[ \mbox{\em setOf}\bigl(\emph{add}(e, t)\bigr) = \{e \} \cup \mbox{\em setOf}(t) \]
\noindent where {\em setOf} can be defined in an ML-like language as:

\begin{small}
\begin{tabbing}
this \= \kill
\> $\mathsf{fun}$ \= $\mathsf{setOf~t}$ = $\mathsf{case}$ \= $\mathsf{t~of}$ \= $\mathsf{Leaf \Rightarrow \emptyset}~| $\\
\> \> \> \>   $\mathsf{Node(l,e,r) \Rightarrow setOf(l) \cup \{e\} \cup setOf(r)}$
\end{tabbing}
\end{small}

\noindent The work in~\cite{Suter2010DPA,Suter2011SMR} provides a foundation towards reasoning about such formulas.
%
%
The approach allows a wide range of problems to be addressed, because it is parametric in several dimensions: (1) the structure of the data type, (2) the elements stored in the data type, (3) the collection type that is the codomain of the catamorphism, and (4) the behavior of the catamorphism itself.
Thus, it is possible to solve a variety of interesting problems, including:
\begin{itemize}
    \item reasoning about the contents of XML messages,
    \item determining correctness of functional implementations of data types, including queues, maps,
        binary trees, and red-black trees,
    \item reasoning about structure-manipulating functions for data types, such as sort and reverse,
    \item computing bound variables in abstract syntax trees to support reasoning over operational semantics and type systems, and
    \item reasoning about simplifications and transformations of propositional logic.
\end{itemize}

The first class of problems is especially important for {\em guards}, devices that mediate information sharing between security domains according to a specified policy.
Typical guard operations include reading field values in a packet, changing fields in a packet, transforming a packet by adding new fields, dropping fields from a packet, constructing audit messages, and removing a packet from a stream.

\begin{example}
Suppose we have a catamorphism \remDirtyWords~that removes from an XML message $m$ all the words in a given blacklist. Also suppose we want to verify the following idempotent property of the catamorphism: the result obtained after applying the catamorphism to a message $m$ twice is the same as the result obtained after applying the catamorphism to $m$ once. We can write this property as a formula that can be decided by the decision procedure in \cite{Suter2010DPA} as follows:
 \[ \remDirtyWords(m) = \remDirtyWords\bigl(\remDirtyWords(m)\bigr) \]

We can also use the decision procedure to verify properties of programs that manipulate algebraic data structures.  First, we turn the program into {\em verification conditions} that are formulas in our logic (c.f., \cite{Hardin2012GLV}), then use the decision procedure to solve these conditions.  A sample verification condition for the {\em add} function is:
\begin{multline*}
(t_1 = \kw{Node}(t_2, e_1, t_3) \land \setOf(t_4) = \setOf(t_2) \cup \{e_2\}) \implies \\
\setOf(\kw{Node}(t_4,e_1,t_3)) = \setOf(t_1) \cup \{e_2\}
\hspace{0.5cm}\vartriangle\!\!\!\!\!\!
\end{multline*}
\end{example}

The procedure \cite{Suter2010DPA} was proved sound for all
catamorphisms and claimed to be complete for a class of catamorphisms
called {\em sufficiently surjective} catamorphisms, which we will
describe in more detail in Section \ref{section:preliminaries}. The
original algorithm in~\cite{Suter2010DPA} was quite expensive to
compute and required a specialized predicates $M_{p}$ and $S_{p}$ to
be defined separately for each catamorphism and proved correct with
respect to the catamorphism using either a hand-proof or a theorem
prover. In~\cite{Suter2011SMR}, a generalized algorithm for the
decision procedure was proposed, based on unrolling the catamorphism.
This algorithm had three significant advantages over the algorithm
in~\cite{Suter2010DPA}: it was much less expensive to compute, it did
not require the definition of $M_{p}$, and it was claimed to be
complete for all sufficiently surjective catamorphisms.

Unfortunately, both algorithms are {\em incomplete} for some sufficiently surjective catamorphisms.  In~\cite{Suter2010DPA}, the proposed algorithms are incomplete for problems involving finite types and formulas involving inequalities that are non-structural (e.g.: $5 + 3 \neq 8$).  In~\cite{Suter2011SMR}, the proposed algorithm is incomplete because of missing assumptions about the range of the catamorphism function.

In this paper, we propose a complete unrolling-based decision
procedure for catamorphisms that satisfy a {\em generalized sufficient
  surjectivity} condition. We also demonstrate that our unrolling
procedure is complete for sufficiently surjective catamorphisms, given
suitable $S_{p}$ and $M_{p}$ predicates.

We then address two open problems with the previous work \cite{Suter2011SMR}: (1) how many catamorphism unrollings are required in order to prove properties using the decision procedure?  and  (2) when is it possible to combine catamorphisms within a formula in a complete way?  We introduce \emph{monotonic} catamorphisms and prove that our decision procedure is complete with monotonic catamorphisms, and this class of catamorphisms gives a linear unrolling bound for the procedure. While monotonic catamorphisms include all catamorphisms introduced by~\cite{Suter2010DPA,Suter2011SMR}, we show that monotonic catamorphisms are a strict subset of sufficiently surjective catamorphisms.
To answer the second question, we introduce {\em associative} catamorphisms, which can be combined within a formula while preserving completeness results.  These associative catamorphisms have the additional property that they require a very small number of unrollings to solve, and we demonstrate that this behavior explains some of the empirical success in applying catamorphism-based approaches on interesting examples from previous papers \cite{Suter2011SMR,Hardin2012GLV}.

We have implemented the decision procedure in an open-source tool called RADA (\underline{r}easoning about \underline{a}lgebraic \underline{da}ta types), which has been used as a back-end tool in the Guardol system \cite{Hardin2012GLV}. The successful use of RADA in the Guardol project on large-scale guard programs demonstrates that the unrolling approach and the tools are sufficiently mature for use on interesting, real-world applications.

This paper offers the following contributions:
\begin{itemize}
\item We propose an unrolling-based decision procedure for algebraic
  data types with \emph{generalized sufficiently surjective}
  catamorphisms.

\item We provide a corrected proof of completeness for the decision
  procedure with generalized sufficiently surjective catamorphisms.

\item We propose a new class of catamorphisms, called \emph{monotonic}
  catamorphisms, and argue that it is a more intuitive notion than
  generalized sufficient surjectivity. We show that the number of
  unrollings needed for monotonic catamorphisms is linear.

 \item We also define an important subclass of monotonic catamorphisms
   called \emph{associative} catamorphisms and show that an arbitrary
   number of these catamorphisms can be combined in a formula while
   preserving decidability. Another nice property of associative
   catamorphisms is that determining whether a catamorphism function
   is associative can be immediately checked by an SMT solver without
   performing unrolling, so we call these catamorphisms {\em
     detectable}. Finally, associative catamorphisms are guaranteed to
   require an exponentially small number of unrollings to solve.

 \item We describe an implementation of the approach, called RADA,
   which accepts formulas in an extended version of the
   \smtlib\ syntax \cite{BarSTSMT10}, and demonstrate it on a range of
   examples.
\end{itemize}

This paper is an expansion of previous work
in~\cite{Pham2013IUB,PhamRADA13}. It provides a complete and better
organized exposition of the ideas from previous work, and includes
substantial new material, including the new notion of generalized
sufficient surjectivity, a set of revised, full proofs that work for
both the class of sufficiently surjective catamorphisms in
\cite{Suter2010DPA} and the new catamorphism classes in this paper, a
demonstration of the relationship between monotonic and sufficiently
surjective catamorphisms, new implementation techniques in RADA, and
new experimental results.

The rest of this paper is organized as follows.
Section \ref{section:related_work} presents some related work that is closest to ours.
In Section \ref{section:dp_and_monotonic_catas}, we present the unrolling-based decision procedure and prove its completeness.
Section \ref{section:monotonic_cata} presents monotonic catamorphisms.
Section \ref{section:assoc_cata} presents associative catamorphisms.
The relationship between different types of catamorphisms is discussed in Section \ref{section:catas_relationship}.
Experimental results for our approach are shown in Section~\ref{section:experimental_results}.
We conclude this paper in Section \ref{section:conclusion}.

\section{Related Work}
\label{section:related_work}
The most relevant work related to the research in this paper fall in two broad categories: verification tools and decision procedures for algebraic data types.
\subsection{Verification Tools for Algebraic Data Types.}
\label{section:relatedwork_tools_for_adt}
We introduce in this paper a new verification tool called RADA to reason about algebraic data types with catamorphisms.
RADA is described in detail in Section \ref{section:experimental_results} and the algorithms behind it are presented in Sections \ref{section:dp_and_monotonic_catas}, \ref{section:monotonic_cata}, and \ref{section:assoc_cata}.
Besides RADA, there are some tools that support catamorphisms (as well as other functions) over algebraic data types. For example, Isabelle \cite{Nipkow2002IPA}, PVS \cite{Owre1992PPV}, and ACL2 \cite{kaufmann2000computer} provide efficient support for both inductive reasoning and evaluation. Although very powerful and expressive, these tools usually need manual assistance and require substantial expert knowledge to construct a proof. On the contrary, RADA is fully automated and accepts input written in the popular \smtlib\ format \cite{BarSTSMT10}; therefore, we believe that RADA is more suited for non-expert users.

In addition, there are a number of other tools built on top of SMT solvers that have support for data types. One of such tools is Dafny \cite{Leino2010DAP}, which supports many imperative and object-oriented features; hence, Dafny can solve many verification problems that RADA cannot. On the other hand, Dafny does not have explicit support for catamorphisms, so for many problems it requires significantly more annotations than RADA.  For example, RADA can, without any annotations other than the specification of correctness, demonstrate the correctness of insertion and deletion for red-black trees.  From examining proofs of similarly complex data structures (such as the PriorityQueue) provided in the Dafny distribution, it is likely that these proofs would require significant annotations in Dafny.


Our work was inspired by the Leon system \cite{Blanc2013OLV}, which uses a semi-decision procedure to reason about catamorphisms \cite{Suter2011SMR}.
While Leon uses Scala input, RADA offers a neutral input format, which is a superset of \smtlib.
Also, Leon specifically uses Z3 \cite{DeMoura2008ZES} as its underlying SMT solver, whereas RADA is solver-independent: it currently supports both Z3 and CVC4. In fact, RADA can support any SMT solver that uses \smtlib\ and that has support for algebraic data types and uninterpreted functions.  RADA also guarantees the completeness of the results even when the input formulas have multiple catamorphisms for certain classes of catamorphisms such as PAC catamorphisms~\cite{Pham2013PAC}; in this situation, it is unknown whether the decision procedure \cite{Suter2011SMR} used in Leon can ensure the completeness.\footnote{The authors of \cite{Suter2011SMR} only claim completeness of the procedure when there is only one non-parametric catamorphism in the input formulas}  Recent work by the Leon group~\cite{Reynolds15} broadens the class of formulas that can be solved by the tool towards arbitrary recursive functions, but it makes no claims on completeness.

%
%
%
\subsection{Decision Procedures for Algebraic Data Types.}
The general approach of using abstractions to summarize algebraic data types has been used in the Jahob system \cite{Zee2008FFV,Zee2009IPL} and
in some procedures for algebraic data types \cite{Sofronie-Stokkermans2009LRC,Suter2011SMR,Jacobs2011TCR,Madhusudan2012RPI}.
However, it is often challenging to directly reason about the abstractions.
One approach to overcome the difficulty (e.g., in \cite{Suter2011SMR,Madhusudan2012RPI}) is to approximate the behaviors of the abstractions using uninterpreted functions
and then send the functions to SMT solvers~\cite{DeMoura2008ZES,Barrett2011CVC4} that have built-in support for uninterpreted functions and recursive data types.

Our approach extends the work by Suter et al. \cite{Suter2010DPA,Suter2011SMR}.
In \cite{Suter2010DPA}, the authors propose a family of procedures for algebraic data types where catamorphisms are used to abstract tree terms.  These procedures are claimed to be sound for all catamorphisms and complete with {\em sufficiently surjective} catamorphisms.  Unfortunately, there are flaws in the completeness argument, and in fact the family of algorithms is incomplete for non-structural disequalities and catamorphisms over finite types.  These incompletenesses and possible fixes to them are described in detail in~\cite{HungPham-PhD}.
An improved approach using a single unrolling-based decision procedure is proposed in~\cite{Suter2011SMR}.  This approach is very similar to the algorithm that is proposed in this paper.  Our approach addresses an incompleteness in the unrolling algorithm due to the use of uninterpreted functions without range restrictions that is described in Section~\ref{section:revised_unrolling_procedure}.

Another similar work is that of Madhusudan et al. \cite{Madhusudan2012RPI},
where a sound, incomplete, and automated method is proposed to achieve recursive proofs for inductive tree data-structures while still maintaining a balance between expressiveness and decidability.
The method is based on {\dryad}, a recursive extension of the first-order logic.
{\dryad} has some limitations: the element values in {\dryad} must be of type \kw{int} and only four classes of abstractions are allowed in {\dryad}.
In addition to the sound procedure, \cite{Madhusudan2012RPI} shows a decidable fragment of verification conditions that can be expressed in {\strand}$_{dec}$ \cite{Madhusudan2011DLC}.
However, this decidable fragment does not allow us to reason about some important properties such as the height or size of a tree.
On the other hand, the class of data structures that \cite{Madhusudan2012RPI} can work with is richer than that of our approach and can involve mutual references between elements (pointers).

Sato et al. \cite{Sato2013TSS} proposes a verification technique that has support for recursive data structures.
The technique is based on higher-order model checking, predicate abstraction, and counterexample-guided abstraction refinement. Given a program with recursive data structures, they encode the structures as functions on lists,
which are then encoded as functions on integers before sending the resulting program to the verification tool described in \cite{Kobayashi2011PAC}. Their method can work with higher-order functions while ours cannot.
On the other hand, their method is incomplete and cannot verify some properties of recursive data structures while ours can thanks to the use of catamorphisms.  An example of such a property is as follows:
{\em after inserting an element to a binary tree,
the set of all element values in the new tree must be a super set of that of the original tree}.

Zhang et al. in~\cite{Zhang04decisionprocedures} define an approach for reasoning over datatypes with integer constraints related to the size of recursive data structures.  This approach is much less general than ours: the size relation in~\cite{Zhang04decisionprocedures} can be straightforwardly constructed as a monotonic integer catamorphism matching the shape of the datatype.  On the other hand, the work in~\cite{Zhang04decisionprocedures} presents a decision procedure for quantified formulas, whilst our approach only supports quantifier-free formulas.

\section{Unrolling-based Decision Procedure}
\label{section:dp_and_monotonic_catas}

Inspired by the decision procedures for algebraic data types by Suter
et al.~\cite{Suter2010DPA,Suter2011SMR}, in this section we present an
unrolling-based decision procedure, the idea of generalized sufficient
surjectivity, and proofs of soundness and completeness of the
procedure for catamorphisms satisfying the condition.

\subsection{Preliminaries}
\label{section:preliminaries}

We describe the parametric logic used in the decision procedures for
algebraic data types, which is also the logic used in our decision
procedure. We also summarize the definition of catamorphisms and the
idea of sufficient surjectivity from~\cite{Suter2010DPA,Suter2011SMR}.
Although the logic and unrolling procedure is parametric with respect
to data types, in the sequel we focus on binary trees to illustrate
the concepts and proofs.

\subsubsection{Parametric Logic}
\label{section:parametric_logic}

The input to the decision procedures is a formula $\phi$ of literals
over elements of tree terms and abstractions produced by a
catamorphism. The logic is {\em parametric} in the sense that we
assume a data type $\rdttype$ to be reasoned about, a decidable
element theory $\logic_\elem$ of values in an element domain $\elem$
containing terms $E$, a catamorphism $\alpha$ that is used to abstract
the data type, and a decidable theory $\logic_\collection$ of values
in a collection domain $\collection$ containing terms $C$ generated by
the catamorphism function. Fig.~\ref{fig:syntax} shows the syntax of
the logic instantiated for binary trees. 
\ref{fig:semantics}. Its semantics can be found in
Fig.~\ref{fig:semantics}. The semantics refer to the catamorphism
$\alpha$ as well as the semantics of elements $[\ ]_{\elem}$ and
collections $[\ ]_{\collection}$. In a slight abuse of notation, we
will also refer to terms in the union of $C$ and $E$ as {\em CE} terms
(respectively, elements of the $\collection \elem$ domain).

\begin{figure}[htb]
\centering
\begin{tabular}{rclr}
\toprule 
$T$    & $::=$ &  $t~|~\Leaf~|~\Node(T,E,T)~|~\kw{left}(T)~|~\kw{right}(T)$ & Tree terms\\
$C$    & $::=$ &  $c~|~\alpha(T)~|~\theory_\collection$ & $\collection$-terms\\
$E$    & $::=$ &  $e~|~\kw{elem}(T)~| ~\theory_\elem$ & $\elem$-terms\\
$F_T$  & $::=$ &  $T = T ~|~ T \neq T$ & Tree (dis)equations \\
$F_C$  & $::=$ &  $C = C ~|~ \func_\collection$& Formula of $\logic_\collection$\\
$F_E$  & $::=$ &  $E = E ~|~ \func_\elem$& Formula of $\logic_\elem$\\
$\phi$ & $::=$ &  $F_T ~|~ F_C ~|~ F_E ~|~ \neg\phi ~|~ \phi\vee\phi ~|~ \phi\wedge\phi ~|~\phi\Rightarrow\phi ~|~ \phi\Leftrightarrow\phi$ & Formulas\\
\bottomrule
\end{tabular}
\caption{Syntax of the parametric logic} 
\label{fig:syntax}
\end{figure}

\begin{figure}[htb]
\centering
\begin{tabular}{rcl}
\toprule
$\mybrack{\Node(T_1, e, T_2)}$                      & $=$ &  $\Node(\mybrack{T_1}, \mybrack{e}_\elem, \mybrack{T_2})$\\
$\mybrack{\Leaf}$                                   & $=$ &  $\Leaf$\\
$\mybrack{\kw{left}(\Node(T_1, e, T_2))}$           & $=$ &  $\mybrack{T_1}$\\
$\mybrack{\kw{right}(\Node(T_1, e, T_2))}$          & $=$ &  $\mybrack{T_2}$\\
$\mybrack{\kw{elem}(\Node(T_1, e, T_2))}$          & $=$ &  $\mybrack{e}_\elem$\\
$\mybrack{\alpha(t)}$                               & $=$ &  given by the catamorphism\\
$\mybrack{ T_1 = T_2}$                              & $=$ &  $\mybrack{T_1} = \mybrack{T_2}$\\
$\mybrack{ T_1 \neq T_2}$                           & $=$ &  $\mybrack{T_1} \neq \mybrack{T_2}$\\
$\mybrack{ E_1 = E_2}$                              & $=$ &  $\mybrack{E_1}_\elem = \mybrack{E_2}_\elem$\\
$\mybrack{\func_\elem}$                       & $=$ &  $\mybrack{\func_\elem}_\elem$\\
$\mybrack{ C_1 = C_2}$                              & $=$ &  $\mybrack{C_1}_\collection = \mybrack{C_2}_\collection$\\
$\mybrack{\func_\collection}$                       & $=$ &  $\mybrack{\func_\collection}_\collection$\\
$\mybrack{\neg\phi}$                                & $=$ &  $\neg\mybrack{\phi}$\\
$\mybrack{\phi_1\star\phi_2}$                       & $=$ &  $\mybrack{\phi_1}\star\mybrack{\phi_2}$ where $\star \in \{\vee, \wedge, \Rightarrow, \Leftrightarrow\}$\\
\bottomrule
\end{tabular}
\caption{Semantics of the parametric logic}
\label{fig:semantics}
\end{figure}

The syntax of the logic ranges over data type terms $T$ and
$\collection$-terms of a decidable collection theory
$\logic_\collection$. $\theory_\collection$ and $\func_\collection$
are arbitrary terms and formulas in $\logic_\collection$, as are
$\theory_\elem$ and $\func_\elem$ in $\logic_\elem$. Tree formulas
$F_T$ describe equalities and disequalities over tree terms.
Collection formulas $F_C$ and element formulas $F_E$ describe
equalities over collection terms $C$ and element terms $E$, as well as
other operations ($\func_\collection$, $\func_\elem$) allowed by the
logic of collections $\logic_\collection$ and elements $\logic_\elem$.
$E$ defines terms in the element types $\elem$ contained within the
branches of the data types. $\phi$ defines formulas in the parametric
logic.

\subsubsection{Catamorphisms}
Given a tree in the parametric logic shown in Fig. \ref{fig:syntax},
we can map the tree to a value in $\collection$ using a {\em catamorphism},
which is a fold function of the following format:
\[
  \alpha(t) =
    \begin{cases}
    \kw{empty}  &\text{if $t = \Leaf$}\\
    \kw{combine}\bigl(\alpha(t_L), e, \alpha(t_R)\bigr) &\text{if $t = \Node(t_L, e, t_R)$}
    \end{cases}
\]
where $\kw{empty}$ is an element in $\collection$
and $\kw{combine}: (\collection, \elem, \collection) \rightarrow \collection$
is a function that combines a triple of two values in $\collection$ and an element in $\elem$
into a value in $\collection$.


\begin{table}[htb]
\caption{Sufficiently surjective catamorphisms in \cite{Suter2010DPA}}
\scriptsize
\begin{tabular}{lclc}
\toprule
Name & $\alpha(\Leaf)$ & $\alpha(\Node(t_L, e, t_R))$ & Example\\
\hline
\emph{Set} & $\emptyset$ & $\alpha(t_L) \cup \{e\} \cup \alpha(t_R)$ & \{1, 2\}\\
\hline
\emph{Multiset} & $\emptyset$ & $\alpha(t_L) \uplus \{e\} \uplus \alpha(t_R)$ & \{1, 2\}\\
\hline
\emph{SizeI} & 0 & $\alpha(t_L) + 1 + \alpha(t_R)$ & 2\\
\hline
\emph{Height} & 0 & 1 + $\max\{\alpha(t_L), \alpha(t_R)\}$ & 2\\
\hline
\multirow{3}{*}{\emph{List}} & \multirow{3}{*}{\List()} & $\alpha(t_L)~@~\List(e)~@~\alpha(t_R)$ (in-order) & (1 2)\\
& & $\List(e)~@~\alpha(t_L)~@~\alpha(t_R)$ (pre-order) & (2 1)\\
& & $\alpha(t_L)~@~\alpha(t_R)~@~\List(e)$ (post-order) & (1 2)\\
\hline
\emph{Some} & \None & $\Some(e)$ & $\Some(2)$\\
\hline
\emph{Min} & \None & $\min'\{\alpha(t_L), e, \alpha(t_R)\}$ & 1\\
\hline
\multirow{2}{*}{\emph{Sortedness}} & \multirow{2}{*}{(\None, \None, \true)} & (\None, \None, \false) (if tree unsorted) & \multirow{2}{*}{(1, 2, \true)}\\
&  & (min element, max element, \true) (if tree sorted) & \\
\bottomrule
\end{tabular}
\label{table:POPL_catas}
\end{table}

The catamorphisms defined in~\cite{Suter2010DPA} are shown in Table \ref{table:POPL_catas}.
The first column contains catamorphism names\footnote{\emph{SizeI}, which maps a tree to its number of \emph{internal} nodes,
was originally named \emph{Size} in \cite{Suter2010DPA}.
We rename the catamorphism to easily distinguish it from the function $size$, which returns the total number of \emph{all} vertices in a tree, in this paper.}.
The next two columns define $\alpha(t)$
when $t$ is a $\Leaf$ and when it is a {\Node}, respectively.
The last column shows examples of the application of each catamorphism to
the tree
in Fig. \ref{fig:tree_example}.


\begin{wrapfigure}[7]{r}{0.4\textwidth}
\includegraphics[width=0.4\textwidth]{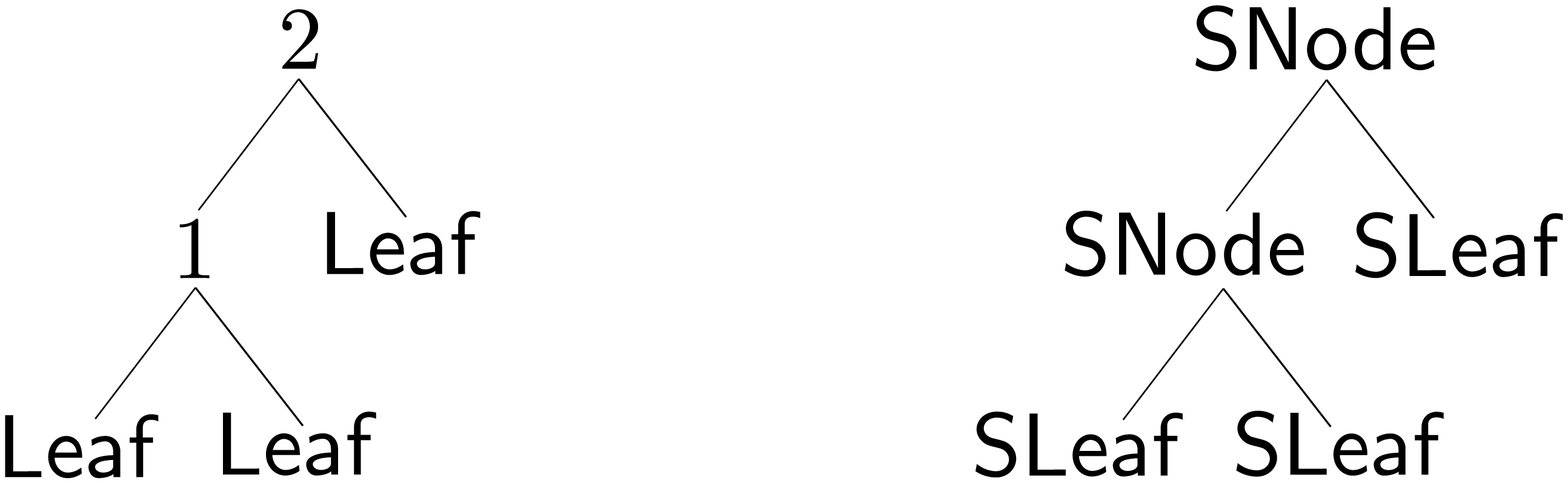}
\caption{An example of a tree and its shape
}
\label{fig:tree_example}
\end{wrapfigure}

In the \emph{Min} catamorphism, $\min'$ is the same as the usual $\min$ function
except that $\min'$ ignores {\None} in the list of its arguments, which must contain at least one non-{\None} value.
The \emph{Sortedness} catamorphism returns a triple containing the min and max element of a tree, and $\true/\false$ depending
on whether the tree is sorted or not.

\paragraph{Infinitely surjective catamorphisms:} Suter et al. \cite{Suter2010DPA} showed that many interesting catamorphisms are \emph{infinitely surjective}. Intuitively, a catamorphism is infinitely surjective if the cardinality of its inverse function is infinite for all but a finite number of trees.
\begin{definition}[Infinitely surjective catamorphisms] A catamorphism $\alpha$ is an infinitely surjective $S$-abstraction, where $S$ is a finite set of trees, if and only if the inverse image $\alpha^{-1}\bigl(\alpha(t)\bigr)$ is finite for $t \in S$ and infinite for $t\notin S$.
\label{definition:infinitely_surjective_catamorphisms}
\end{definition}
\begin{example}[Infinitely surjective catamorphisms]
\label{example:infinitely_surjective_catamorphisms}
The \emph{Set} catamorphism in Table~\ref{table:POPL_catas} is an infinitely surjective \{\Leaf\}-abstraction because:
\begin{itemize}
 \item $\bigl|\emph{Set}^{-1}\bigl(\emph{Set}(\Leaf)\bigr)\bigr| = |\emph{Set}^{-1}(\emptyset)| = 1$ (i.e., $\Leaf$ is the only tree in data type $\rdttype$ that can map to $\emptyset$ by the \emph{Set} catamorphism). Hence, $\emph{Set}^{-1}\bigl(\emph{Set}(\Leaf)\bigr)$ is finite.
 \item $\forall t \in \rdttype, t \neq \Leaf: \bigl|\emph{Set}^{-1}\bigl(\emph{Set}(t)\bigr)\bigr| = \infty$.
 The reason is that when $t$ is not $\Leaf$, we have $\emph{Set}(t) \neq \emptyset$. Hence, there are an infinite number of trees that can map to $\emph{Set}(t)$ by the catamorphism \emph{Set}. For example, consider the tree in Fig. \ref{fig:tree_example}; let us call it $t_0$. We have $\emph{Set}(t_0) = \{1, 2\}$; hence, $|\emph{Set}^{-1}(\{1, 2\})| = \infty$ since there are an infinite number of trees in $\rdttype$ whose elements values are 1 and 2.
 \end{itemize}
As a result, \emph{Set} is infinitely surjective by Definition \ref{definition:infinitely_surjective_catamorphisms}.
\exampleEndMark
\end{example}

\paragraph{Sufficiently surjective catamorphisms:}
The decision procedures by Suter et
al.~\cite{Suter2010DPA,Suter2011SMR} were claimed to be complete if
the catamorphism used in the procedures is \emph{sufficiently
  surjective} \cite{Suter2010DPA}. Intuitively, a catamorphism is
sufficiently surjective if the inverse of the catamorphism has
sufficiently large cardinality for all but a finite number of tree
shapes. In fact, the class of infinitely surjective catamorphisms is
just a special case of sufficiently surjective catamorphisms
\cite{Suter2010DPA}.

To define the notion of sufficiently surjective catamorphisms, we have to define \emph{tree shapes} first. The shape of a tree is obtained by removing all element values in the tree. Fig. \ref{fig:tree_example} shows an example of a tree and its shape.
\begin{definition}[Tree shapes]
\label{definition:shapes}
The shape of a tree is defined by constant $\SLeaf$ and constructor $\SNode(\_, \_)$ as follows:
\[
  \shape(t) =
    \begin{cases}
    \SLeaf  &\text{if $t = \Leaf$}\\
    \SNode\bigl(\shape(t_L), \shape(t_R)\bigr) &\text{if $t = \Node(t_L, \_, t_R)$}\\
    \end{cases}
\]
\end{definition}

\begin{definition}[Sufficiently surjective catamorphisms \cite{Suter2010DPA}]
\label{definition:sufficient_surjectivity}
A catamorphism $\alpha$ is sufficiently surjective iff for each $p \in \naturalnums^+$,
there exists, computable as a function of $p$,
\begin{itemize}
 \item a finite set of shapes $S_p$
 \item a closed formula $M_p$ in the union of the collection and element 
   theories\footnote{Note that Suter et. al in~\cite{Suter2010DPA} describe $M_{p}$ over the collection theory only, but that paper contains examples that involve both the collection and element theory (c.f., $M_{p}$ for multiset catamorphisms).  The addition of the element theory does not require modification to any of the proofs in our work or~\cite{Suter2010DPA}.} such that for any collection element $c$, $M_p(c)$ implies $|\alpha^{-1}(c)| > p$
\end{itemize}
such that  $M_p\bigl(\alpha(t)\bigr)$ or $\shape(t) \in S_p$ for every tree term $t$.
\end{definition}

\begin{example}[Sufficiently surjective catamorphisms]
We showed in Example~\ref{example:infinitely_surjective_catamorphisms}
that the \emph{Set} catamorphism is infinitely surjective. Let us now
show that the catamorphism is sufficiently surjective by Definition
\ref{definition:sufficient_surjectivity}. Let $S_p = \{\SLeaf\}$ and
$M_p(c) \equiv c \neq \emptyset$. For this $M_p$, the only base case
to consider is the tree $\Leaf$: either a tree is $\Leaf$, whose shape
is in $S_p$, or the catamorphism value returned is not the empty set,
in which case $M_p$ holds. Furthermore, $M_p(c)$ implies
$|\alpha^{-1}(c)| = \infty$. \exampleEndMark
\end{example}

Despite its name, sufficient surjectivity has no surjectivity requirement for the range of $\alpha$.
It only requires a ``sufficiently large" number of trees for values satisfying the condition $M_p$.
The \emph{SizeI} catamorphism is a good example of a sufficiently surjective catamorphism that is not surjective. In other words, there is no restriction for the range of a sufficiently surjective catamorphism.
However, to ensure the completeness of the unrolling decision procedure, the range restriction must be taken into account. We will discuss this issue in Section \ref{section:revised_unrolling_procedure}.

Table \ref{table:POPL_catas} describes all sufficiently surjective catamorphisms in \cite{Suter2010DPA}.
The only catamorphism in \cite{Suter2010DPA} not in Table \ref{table:POPL_catas} is the \Mirror~catamorphism:
\[
  \Mirror(t) =
    \begin{cases}
    \Leaf  &\text{if $t = \Leaf$}\\
    \Node\bigl(\Mirror(t_R), e, \Mirror(t_L)\bigr) &\text{if $t = \Node(t_L, e, t_R)$}\\
    \end{cases}
\]
\noindent 
Since the cardinality of the inversion function of the catamorphism \Mirror~is always 1,
the sufficiently surjective condition does not hold for this catamorphism.



\subsection{Properties of Trees and Shapes in the Parametric Logic}
\label{section:properties}
We present some important properties of trees and shapes in the
parametric logic (Section \ref{section:parametric_logic}) which play
important roles in the subsequent sections of this paper.

\paragraph{\underline{Properties of Trees.}}
We assume the standard definitions of height and size for trees in the
parametric logic with $\height(\Leaf) = 0$ and $\size(\Leaf) = 1$. The
following properties result directly from structural induction on
trees in the parametric logic.

\begin{property}[Type of tree]
\label{property:type_of_tree}
Any tree in the parametric logic is a full binary tree.
\end{property}

\begin{property}[Size]
\label{property:size_is_odd}
The number of vertices in any tree in the parametric logic is odd.
Also, in a tree $t$ of size $2k + 1$ ($k \in \naturalnums$), we have
$k$ internal nodes and $k+1$ leaves.
\end{property}

\begin{property}[Size vs. Height]
\label{property:size_height}
In the parametric logic, the size of a tree of height $h \in
\naturalnums$ must be at least $2h + 1$:
\[\forall t \in \rdttype : \size(t) \geq 2 \times \height(t) + 1\]
\end{property}

\paragraph{\underline{Properties of Tree Shapes.}}
We now show a special relationship between tree shapes and the well-known Catalan numbers \cite{Stanley2001},
which, according to Koshy \cite{Koshy2009},
can be computed as follows:
\begin{align*}
\catalan_0 &= 1 &
\catalan_{n + 1} &= \frac{2(2n + 1)}{n + 2}\catalan_n~\mbox{(where $n \in \naturalnums$)}
\end{align*}
where $\catalan_n$ is the $n^\text{th}$ Catalan number. Catalan numbers will
be used to establish some properties of associative catamorphisms in
Section \ref{section:assoc_cata}.


Define the size of the shape of a tree to be the size of the tree.
Let $\oddnums$ be the set of odd natural numbers.
Due to Property \ref{property:size_is_odd}, the size of a shape is in $\oddnums$.
Let $\numshapes(s)$ be the \underline{n}umber of tree \underline{s}hapes of size $s \in \oddnums$.

\begin{lemma}
\label{lemma:num_shapes_catalan}
The number of shapes of size $s \in \oddnums$ is the $\frac{s-1}{2}$-th Catalan number:
\[\numshapes(s) = \catalan_{\frac{s-1}{2}}\]
\end{lemma}
\begin{proof}
Property \ref{property:type_of_tree} implies that tree shapes are also full binary trees.
The lemma follows since the number of full binary trees of size $s \in \oddnums$ is $\catalan_{\frac{s-1}{2}}$ \cite{Stanley2001,Koshy2009}.
\qed
\end{proof}

\begin{lemma}
\label{lemma:numshapes} Function $\numshapes: \oddnums \rightarrow \naturalnums^+$ is monotone:
\[1 = \numshapes(1) = \numshapes(3) < \numshapes(5) < \numshapes(7) < \numshapes(9) <\ldots\]
\end{lemma}
\begin{proof}
Clearly, $\catalan_1 = \catalan_0 = 1$. When $n \geq 1$, we have:
\begin{equation*}
\catalan_{n + 1} = \frac{2(2n + 1)}{n + 2}\catalan_n > \frac{2(2n +
  1)}{4n + 2}\catalan_n = \catalan_n
\end{equation*}
Therefore, by induction on $n$, we obtain: $1 = \catalan_0 =
\catalan_1 < \catalan_2 < \catalan_3 < \catalan_4 < \ldots$, which
completes the proof because of Lemma \ref{lemma:num_shapes_catalan}.
\qed
\end{proof}


\subsection{Unrolling-based Decision Procedure Revisited}
\label{section:revised_unrolling_procedure}

This section presents our unrolling-based decision procedure, which was inspired by the work by Suter et al. \cite{Suter2011SMR}.
First, let us define two notions that will be used frequently throughout the discussions in this section.

\emph{An uninterpreted function representing catamorphism
  applications.} The evaluation of $\alpha(t_0)$ for some tree term
$t_0 \in \rdttype$ might depend on the value of some $\alpha(t_0')$
that we have no information to evaluate. In this case, our decision
procedure treats $\alpha(t_0')$ as an application of the uninterpreted
function $\uf(t_0')$, where $\uf: \rdttype \rightarrow \collection$.

%
%
%
For example, suppose that only $\alpha(\kw{left}(t_0))$ needs to be considered as an uninterpreted function while evaluating $\alpha(t_0)$; we can compute $\alpha(t_0)$ as follows:
\[
  \alpha(t_0) =
    \begin{cases}
    \kw{empty}  &\text{if $t_0 = \Leaf$}\\
    \kw{combine}\bigl(\uf\bigl(\kw{left}(t_0)\bigr), \kw{elem}(t_0), \alpha\bigl(\kw{right}(t_0)\bigr)\bigr) &\text{if $t_0 \neq \Leaf$}
    \end{cases}
\]

\emph{Control conditions.} For each catamorphism application $\alpha(t)$, we use a control condition $b_t$ to check whether the evaluation of $\alpha(t)$ depends on the uninterpreted function $\uf$ or not. If $b_t$ is $\true$,
we can evaluate $\alpha(t)$ without calling to the uninterpreted function $\uf$.
For the $\alpha(t_0)$ example above, we have $b_{t_0} \equiv t_0 = \Leaf$.

The unrolling procedure proposed by Suter et al. \cite{Suter2011SMR} is restated in Algorithm~\ref{alg:sas_dp},
and our revised unrolling procedure is shown in Algorithm~\ref{alg:revised_dp}.
%
The input of both algorithms consists of
\begin{itemize}
 \item a formula $\phi$ written in the parametric logic (described in Section \ref{section:parametric_logic}) that consists of literals over elements of tree terms and tree abstractions generated by a catamorphism (i.e., a fold function that maps a recursively-defined data type to a value in a base domain).
In other words,  $\phi$ contains a recursive data type $\rdttype$ (a tree term as defined in the syntax), an element type $\elem$ of the value stored in each tree node, a collection type $\collection$ of tree abstractions in a decidable logic $\logic_\collection$, and a catamorphism $\alpha: \rdttype \rightarrow \collection$ that maps an object in the data type $\rdttype$ to a value in the collection type $\collection$.
 \item a program $\Pi$, which contains $\phi$, the definitions of data type $\rdttype$, and catamorphism $\alpha$.
\end{itemize}

\bigskip
\noindent
\SetInd{0.5em}{0.28em}
\begin{minipage}{0.49\textwidth}
\scriptsize
\begin{algorithm}[H]
\label{alg:sas_dp}
  $\phi \leftarrow \phi[U_\alpha / \alpha]$ \\
  $(\phi, B, \_) \leftarrow \unrollStep(\phi, \Pi, \emptyset)$\\
  \While{\true} {
    \Switch{decide($\phi \wedge \bigwedge_{b \in B} b$)}{
      \Case{SAT}
      {
        \Return{``SAT''} \label{dpcode:sat_case}
      }
      \Case {UNSAT}
      {
        \Switch{decide($\phi$)}{ \label{dpcode:decide_without_control_conditions}
          \Case{UNSAT}
          {
            \Return{``UNSAT''} \label{dpcode:unsat_case}
          }
          \Case{SAT}
          {
            $(\phi, B, \_ ) \leftarrow $ \\
            \quad $\unrollStep(\phi, \Pi, B)$
          } \label{dpcode:sat_without_control_conditions}
        }
      }
    }
  }
\caption{Unrolling decision procedure in \cite{Suter2011SMR} with \emph{sufficiently surjective} catamorphisms}
\end{algorithm}
\end{minipage}
\hspace{0.1cm}
\begin{minipage}{0.49\textwidth}
\begin{algorithm}[H]
\label{alg:revised_dp}
\scriptsize
  $\phi \leftarrow \phi[U_\alpha / \alpha]$ \\
  $(\phi, B, R) \leftarrow \unrollStep(\phi, \Pi, \emptyset)$\\
  \While{\true} {
    \Switch{decide($\phi \wedge \bigwedge_{b \in B} b$)}{
      \Case{SAT}
      {
        \Return{``SAT''} \label{reviseddpcode:sat_case}
      }
      \Case {UNSAT}
      {
        \Switch{decide($\phi \wedge \bigwedge_{r \in R} r) $}{\label{reviseddpcode:decide_without_control_conditions}
          \Case{UNSAT}
          {
            \Return{``UNSAT''} \label{reviseddpcode:unsat_case}
          }
          \Case{SAT}
          {
            $(\phi, B, R) \leftarrow $ \\
            \quad $\unrollStep(\phi, \Pi, B)$\label{reviseddpcode:unsat_case_end}
          } \label{reviseddpcode:sat_without_control_conditions}
        }
      } \label{reviseddpcode:unsat_case_start}
    }
  }
\caption{Unrolling decision proc. with \emph{generalized sufficiently surjective} catamorphisms (Def. \ref{definition:generalized_sufficient_surjectivity})}
\end{algorithm}
\end{minipage}
\smallskip

The decision procedure works on top of an SMT solver
$\solver$ that supports theories for $\rdttype, \elem, \collection$, and uninterpreted functions.
Note that the only part of the parametric logic that is not inherently supported by $\solver$ is the applications of the catamorphism.
The main idea of the decision procedure is to approximate the behavior of the catamorphism by repeatedly unrolling it
and treating the calls to the not-yet-unrolled catamorphism instances at the leaves as calls to an uninterpreted function $U_{\alpha}$.
We start by replacing all instances of the catamorphism
$\alpha$ by instances of an uninterpreted function $U_{\alpha}$ using
the substitution notation $\phi[U_\alpha / \alpha]$.
The uninterpreted function can return any values in its codomain;
thus, the presence of this uninterpreted function can make \emph{SAT} results untrustworthy.
To address this issue, each time the catamorphism is unrolled, a set of boolean control conditions $B$ is created to
determine if the determination of satisfiability is independent of the uninterpreted function at the ``leaf'' level of the unrolling.
That is, if all the control conditions in $B$ are true, the uninterpreted function $U_{\alpha}$ does not play any role in the satisfiability result.
The unrollings without control conditions represent an
over-approximation of the formula with the semantics of the program
with respect to the parametric logic, in that it accepts all models
accepted by the logic plus some others (due to the
uninterpreted function).  The unrollings with control conditions
represent an under-approximation: all models accepted by this model will
be accepted by the logic with the
catamorphism.


In addition, we observe that if a catamorphism instance is treated as
an uninterpreted function, the uninterpreted function should only
return values inside the \range\footnote{The codomain of a
  function $f: X \rightarrow Y$ is the set $Y$ while the range of $f$
  is the actual set of all of the output the function can return in
  $Y$. For example, the codomain of \emph{Height} when defined against
  \smtlib~\cite{BarSTSMT10} is $\kw{Int}$ while its range is the set
  of natural numbers.} of the catamorphism. In our decision procedure,
a user-provided predicate $R_\alpha$ captures the range constraint of
the catamorphism. $R_\alpha$ is applied to instances of $\uf(t)$ to
constrain the values of the uninterpreted function to the range of
$\alpha$.

\begin{algorithm}[htb]
\label{alg:unroll_step}
  \If(\tcc*[f]{Function called for the first time}){$B = \emptyset$ } {
    $\unrolledtrees \leftarrow \{ t~|~\alpha(t) \in \phi \} $  \tcc*[f]{Global set of frontier nodes} \\ 
    $B \leftarrow \{ \kw{false} \}$ \\ 
    $R \leftarrow \{R_\alpha(\uf(t))~|~t \in \unrolledtrees \} $ \\ 
  }
  \Else{
    $\unrolledtrees \leftarrow \bigcup \{ \{\kw{left}(t), \kw{right}(t)\}~|~t \in \unrolledtrees \} $ \\
  $B \leftarrow \emptyset$ \\
  $R \leftarrow \emptyset$ \\
  \For{$t \in \unrolledtrees$} {
    $B \leftarrow B \cup \{ t = \kw{Leaf} \}$ \\
    $\phi \leftarrow \phi \land (\kw{$\uf$}(t) = $ \\
      \ \ \ \ \ \ \ \ \ \ \ \ \ \ \ \
      $(\kw{ite}~(t = \kw{Leaf})~
         \kw{empty}_{\Pi}~
                (\kw{combine}_{\Pi}(\uf(\kw{left}(t)), \kw{elem}(t), \uf(\kw{right}(t))))))$ \\
    $R \leftarrow R \cup \{R_\alpha(\uf(\kw{left}(t))),
     R_\alpha(\uf(\kw{right}(t)))\}$
  }
  }
  \Return {$\phi, B, R$}
\caption{Algorithm for $\unrollStep(\phi, \Pi, B)$}
\end{algorithm}

Algorithm \ref{alg:revised_dp} determines the satisfiability of $\phi$ through repeated unrollings of $\alpha$ using the \unrollStep~ function in Algorithm \ref{alg:unroll_step}.
Given a formula $\phi_i$ generated from the original $\phi$ after unrolling the catamorphism $i$ times and the set of control conditions $B_i$ of $\phi_i$, function $\unrollStep(\phi_i, \Pi, B_i)$ unrolls the catamorphism one more time and returns a triple $(\phi_{i+1}, B_{i+1}, R_{i+1})$ containing the unrolled version $\phi_{i+1}$ of $\phi_i$, a set of control conditions $B_{i+1}$ for $\phi_{i+1}$, and a set of range restrictions $R_{i+1}$ for elements of the codomain of $U_{\alpha}$ corresponding to trees in the leaf-level of the unrolling.

The mechanism by which Algorithm~\ref{alg:unroll_step} ``unrolls'' the catamorphism is actually by constraining the values returned by $\uf$.  This is done with equality constraints that describe the structure of the catamorphism.  Each time we unroll, we start from a set of recently unrolled ``frontier'' vertices $\unrolledtrees$ that define the nodes at the current leaf-level of unrolling. $\unrolledtrees$ is initialized when the function is called for the first time. We then extend the frontier by examining the left and right children of the frontier nodes and define the structure of $\alpha$ over the (previously unconstrained) left and right children of the current frontier of the unrolling process.  The unrolling checks whether or not the tree in question is a $\Leaf$; if so, its value is $\kw{empty}_{\Pi}$; otherwise, its value is the result of applying the approximated catamorphism $\uf$ to its children.

Function $decide(\varphi)$ in Algorithm \ref{alg:revised_dp} simply calls the solver $\solver$ to check the satisfiability of $\varphi$ and returns \emph{SAT}/\emph{UNSAT} accordingly.
Algorithm \ref{alg:revised_dp} either terminates when $\phi$ is proved
to be satisfiable without the use of the uninterpreted function (line \ref{dpcode:sat_case}) or $\phi$ is proved to be unsatisfiable when the presence of uninterpreted function cannot make the problem satisfiable (line \ref{dpcode:unsat_case}).

Let us examine how satisfiability and unsatisfiability are determined in the procedure.
In general, the algorithm keeps unrolling the catamorphism until we find a \emph{SAT}/\emph{UNSAT} result that we can trust.
To do that, we need to consider several cases after each unrolling step is carried out.
First, at line 5, $\phi$ is satisfiable and all the control conditions
are true, which means the uninterpreted function is not involved in the satisfiable result.
In this case, we have a complete tree model for the \emph{SAT} result and we can conclude that the problem is satisfiable.

On the other hand, consider the case when $decide(\phi \wedge
\bigwedge_{b \in B} b) = \emph{UNSAT}$. The \emph{UNSAT} may be due to
the unsatisfiability of $\phi$, or the set of control conditions, or
both of them together. To understand the \emph{UNSAT} case more deeply, we
could try to check the satisfiability of $\phi$ alone. Note that
checking $\phi$ alone would mean that the control conditions are not
used; consequently, the values of the uninterpreted function could
contribute to the \emph{SAT}/\emph{UNSAT} result. Therefore, we
instead check $\phi$ with the range restrictions on the uninterpreted
function in the satisfiability check (i.e., $decide(\phi \wedge
\bigwedge_{r \in R} r)$ at line 8) to ensure that if a catamorphism
instance is viewed as an uninterpreted function then the uninterpreted
function only returns values inside the range of the catamorphism. If
$decide(\phi \wedge \bigwedge_{r \in R} r) = \emph{UNSAT}$ as at line
9, we can conclude that the problem is unsatisfiable because the
presence of the uninterpreted function still cannot make the problem
satisfiable as a whole. Finally, we need to consider the case
$decide(\phi \wedge \bigwedge_{r \in R} r) = \emph{SAT}$ as at line
11. Since we already know that $decide(\phi \wedge \bigwedge_{b \in B}
b) = \emph{UNSAT}$, the only way to make $decide(\phi \wedge
\bigwedge_{r \in R} r)$ = \emph{SAT} is by using at least one value
returned by the uninterpreted function, which also means that the
\emph{SAT} result is untrustworthy. Therefore, we need to keep
unrolling at least one more time as denoted at line 12.

The central problem of Algorithm \ref{alg:sas_dp} is that its termination is not guaranteed.
For example, non-termination can occur
if the uninterpreted function $\uf$
representing $\alpha$ can return values outside the range of $\alpha$.
For example, consider an unsatisfiable formula: $\emph{SizeI}(t) < 0$ when {\em SizeI} is defined over the integers in an SMT solver.
Although \emph{SizeI} is sufficiently surjective \cite{Suter2010DPA}, Algorithm \ref{alg:sas_dp} will not terminate since each uninterpreted function at the leaves of
the unrolling can always choose an arbitrarily large negative number to assign
as the value of the catamorphism, thereby creating a
satisfying assignment when evaluating the input formula without control conditions.
These negative values are outside the range of {\em SizeI}.
Broadly speaking, this termination
problem can occur for any catamorphism that is not surjective.  Unless an
underlying solver supports predicate subtyping, such
catamorphisms are easily constructed. In fact, \emph{SizeI} and \emph{Height} catamorphisms
are not surjective when defined against \smtlib~\cite{BarSTSMT10}.
The issue involves the definition of sufficient surjectivity,
which does not actually
require that a catamorphism be surjective, i.e., defined across the
entire codomain.  All that is required for sufficient surjectivity is
a predicate $M_p$ that constrains the catamorphism value to represent
``acceptably large'' sets of trees.  The \emph{SizeI} catamorphism is an
example of a sufficiently surjective catamorphism that is not
surjective.

To address the non-termination issue, we need to constrain the
assignments to the uninterpreted function $\uf$ representing
$\alpha$ to return only values from the range of $\alpha$. A
user-provided predicate $R_\alpha$ is used as a recognizer for the
range of $\alpha$ to make sure that any values that uninterpreted
function $\uf$ returns can actually be returned by $\alpha$:
\begin{equation}
  \forall c \in \collection: R_\alpha(c) \Leftrightarrow (\exists t
  \in \rdttype : \alpha(t) = c) \label{correctr}
\end{equation}
Formula~\eqref{correctr}~defines a correctness condition for
$R_\alpha$. Unfortunately, it is difficult to prove this without the
aid of a theorem prover. On the other hand, it is straightforward to
determine whether $R_\alpha$ is an overapproximation of the range of
$\alpha$ (that is, all values in the range of $\alpha$ are accepted by
$R_\alpha$) using an inductive argument that can be checked using an SMT
solver. To do so, we check the following formula:
\begin{align*}
\exists c_{1},c_{2} \in \collection, e \in \elem :\ & (\lnot
R_{\alpha}(\kw{empty}_{\Pi})) \ \lor \\
& (R_{\alpha}(c_{1}) \land R_{\alpha}(c_{2}) \land \lnot R_{\alpha}(\kw{combine}_{\Pi}(c_{1}, e, c_{2})))
\end{align*}
This formula, which can be directly analyzed by an SMT solver, checks
whether $R_{\alpha}$ is true for leaf-level trees (checking
$\kw{empty}$) and for non-leaf trees (using an inductive argument over
$\kw{combine}$). If the solver proves that the formula is
\emph{UNSAT}, then $R_{\alpha}$ overapproximates the range of
$\alpha$.

This check ensures that the unrolling algorithm is sound (we don't
`miss' any possible catamorphism values) but not that it is complete.
For example, the predicate $R_{\alpha}(c) = \kw{true}$ recognizes the
entire codomain, $\collection$, and leads to the incompletenesses
mentioned earlier for the $SizeI$ and $Height$ catamorphisms. In our
approach, it is the user's responsibility to provide an accurate range
recognizer predicate.




\subsubsection{Catamorphism Decision Procedure by Example}
\label{section:dp_by_example}
As an example of how the procedure in Algorithm \ref{alg:revised_dp} can be used, let us consider a guard application (such as those in~\cite{Hardin2012GLV}) that needs to determine whether an HTML message may be sent across a trusted to untrusted network boundary.
One aspect of this determination may involve checking whether the message contains a significant number of ``dirty words''; if so, it should be rejected.  Our goal is to ensure that this guard application works correctly.

We can check the correctness of this program by splitting the analysis into two parts.  A verification condition generator (VCG) generates a set of formulas to be proved about the program 
and a back end solver attempts to discharge the formulas.
In the case of the guard application, these back end formulas involve tree terms representing the HTML message, a catamorphism representing the number of dirty words in the tree, and equalities and inequalities involving string constants and uninterpreted functions for determining whether a word is ``dirty''.

In our dirty-word example, the tree elements are strings and we can map a tree to the number of its dirty words by the following $\DW: \rdttype \rightarrow \kw{int}$ catamorphism:
\[
  \DW(t) =
    \begin{cases}
    0 &\text{if $t = \Leaf$}\\
    \DW(t_L) + \bigl({\ite}~\dirty(e)~1~0\bigr) + \DW(t_R)&\text{if $t = \Node(t_L, e, t_R)$}\\
    \end{cases}
\]
where $\elem$ is $\kw{string}$ and $\collection$ is $\kw{int}$. We use ${\ite}$ to denote an if-then-else statement.

For our guard example, suppose one of the verification conditions is:
\[t = \kw{Node}(t_{L}, e, t_{R}) \wedge \dirty(e) \wedge \DW(t) = 0\]
\noindent which is \emph{UNSAT}: since $t$ has at least one dirty word (i.e., value $e$), its number of dirty words cannot be 0. Fig. \ref{fig:dp_example_dw} shows how the procedure works in this case.

\begin{figure}[htb]
   \includegraphics[width=\textwidth]{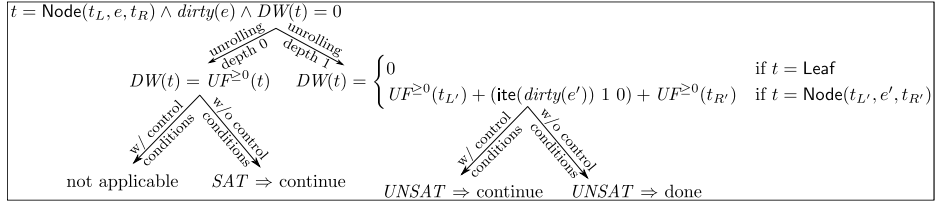}
   \caption{An example of how the decision procedure works}
  \label{fig:dp_example_dw}
\end{figure}

At unrolling depth 0, $\DW(t)$ is treated as an uninterpreted function $\UF^{\geq 0} : \tau \rightarrow \kw{int}$, which, given a tree, can return any value of type $\kw{int}$ (i.e., the codomain of $\DW$) bigger or equal to 0 (i.e., the range of $\DW$). The use of $\UF^{\geq 0}(t)$ implies that for the first step we do not use control conditions. The formula becomes
\[t = \kw{Node}(t_{L}, e, t_{R}) \wedge \dirty(e) \wedge \DW(t) = 0 \wedge \DW(t) = \UF^{\geq 0}(t)\]
and is \emph{SAT}. However, the \emph{SAT} result is untrustworthy due to the presence of $\UF^{\geq 0}(t)$; thus, we continue unrolling $\DW(t)$.

At unrolling depth 1, we allow $\DW(t)$ to be unrolled up to
depth 1 and all the catamorphism applications at lower depths will be
treated as instances of the uninterpreted function. In particular,
$\UF^{\geq 0}(t_{L'})$ and $\UF^{\geq 0}(t_{R'})$ are the
values of the uninterpreted function for $\DW(t_{L'})$ and
$\DW(t_{R'})$, respectively. The set of control conditions in
this case is $\{t = \Leaf\}$; if we use the set of control
conditions (i.e., all control conditions in the set hold), the values
of $\UF^{\geq 0}(t_{L'})$ and $\UF^{\geq 0}(t_{R'})$ will
not be used. Hence, in the case of using the control conditions, the
formula becomes:
\begin{align*}
t = \kw{Node}(t_{L}, e, t_{R}) \wedge \dirty(e) \wedge \DW(t) = 0 \wedge (t = \Leaf) \wedge \Bigl(\DW(t) = 0 \wedge t = \Leaf~\vee \\
\DW(t) = \UF^{\geq 0}(t_{L'}) + \bigl({\textsf{ite}}~\dirty(e')~1~0\bigr) + \UF^{\geq 0}(t_{R'}) \wedge t = \textsf{Node}(t_{L'}, e', t_{R'})\Bigr)
\end{align*}
\noindent which is equivalent to
\[t = \kw{Node}(t_{L}, e, t_{R}) \wedge \dirty(e) \wedge \DW(t) = 0 \wedge t = \Leaf\]
\noindent which is \emph{UNSAT} since $t$ cannot be $\Node$ and $\Leaf$ at the same time. Since we get \emph{UNSAT} with control conditions, we continue the process without using control conditions. Without control conditions, the formula becomes
\begin{align*}
t = \kw{Node}(t_{L}, e, t_{R}) \wedge \dirty(e) \wedge \DW(t) = 0 \wedge \Bigl(\DW(t) = 0 \wedge t = \Leaf~\vee \\
\DW(t) = \UF^{\geq 0}(t_{L'}) + \bigl({\textsf{ite}}~\dirty(e')~1~0\bigr) + \UF^{\geq 0}(t_{R'}) \wedge t = \textsf{Node}(t_{L'}, e', t_{R'})\Bigr)
\end{align*}
\noindent which, after eliminating the $\Leaf$ case (since $t$ must be a $\Node$) and unifying $\kw{Node}(t_{L}, e, t_{R})$ with $\kw{Node}(t_{L'}, e', t_{R'})$ (since they are equal to $t$), simplifies to
\begin{align*}
& t = \kw{Node}(t_{L}, e, t_{R}) \wedge \dirty(e) \wedge \DW(t) = 0~\wedge \\
& \DW(t) = \UF^{\geq 0}(t_{L'}) +
\bigl({\textsf{ite}}~\dirty(e)~1~0\bigr) + \UF^{\geq 0}(t_{R'})
\end{align*}
\noindent which, after evaluating the \kw{ite} expression, is equivalent to
\[t = \kw{Node}(t_{L}, e, t_{R}) \wedge \dirty(e) \wedge \DW(t) = 0 \wedge \DW(t) = \UF^{\geq 0}(t_{L'}) + 1 + \UF^{\geq 0}(t_{R'})\]
\noindent which is \emph{UNSAT} because $\UF^{\geq 0}(t_{L'}) + 1 + \UF^{\geq 0}(t_{R'}) > 0$.
Getting \emph{UNSAT} without
control conditions guarantees that the original formula is \emph{UNSAT}.

We have another example of the procedure in Example \ref{appendix:incremental_solving_with_rada} in Section \ref{section:rada}.


\subsection{Correctness of the Unrolling Decision Procedure}
\label{section:revsed_dp_proof}

We now prove the correctness of the unrolling decision procedure in
Algorithm~\ref{alg:revised_dp}.
First, let us define the notion of the cardinality of the inverse function of catamorphisms.

\begin{definition}[Function $\beta$]
\label{definition:beta}
Given a catamorphism $\alpha : \rdttype \rightarrow \collection$, we define $\beta: \rdttype \rightarrow \naturalnums \cup \{\infty \}$ as the cardinality of the inverse function of $\alpha(t)$:
\[\beta(t) = |\alpha^{-1}\bigl(\alpha(t)\bigr)|\]
\end{definition}

\begin{example}[Function $\beta$]
Intuitively, $\beta(t)$ is the number of distinct trees that can map to $\alpha(t)$ via catamorphism $\alpha$.
The value of $\beta(t)$ clearly depends on $\alpha$.
For example, for the \emph{Set} catamorphism, $\beta_{\emph{Set}}(\Leaf) = 1$; also, $\forall t \in \rdttype, t \neq \Leaf: \beta_{\emph{Set}}(t) = \infty$
since there are an infinite number of trees that have the same set of element values.
For the \DW catamorphism in Section \ref{section:dp_by_example}, we have $\forall t\in \rdttype: \beta_{\DW}(t) = \infty$.

\begin{table}[htb]
\caption{Examples of $\beta(t)$ with the \emph{Multiset} catamorphism}
\footnotesize
\begin{tabular}{cccccc}
\toprule
$t$&~~\includegraphics[scale=0.15]{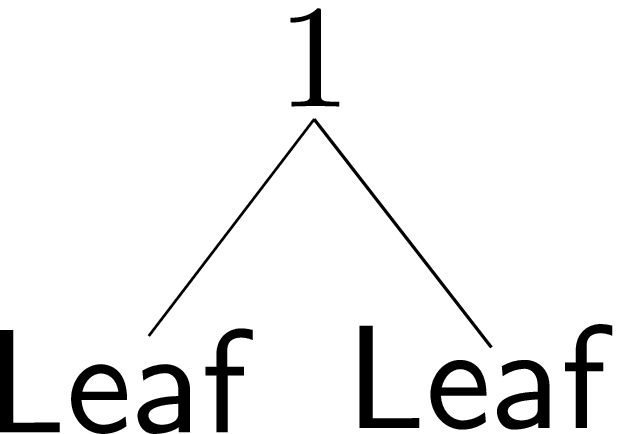}~~&~~\includegraphics[scale=0.15]{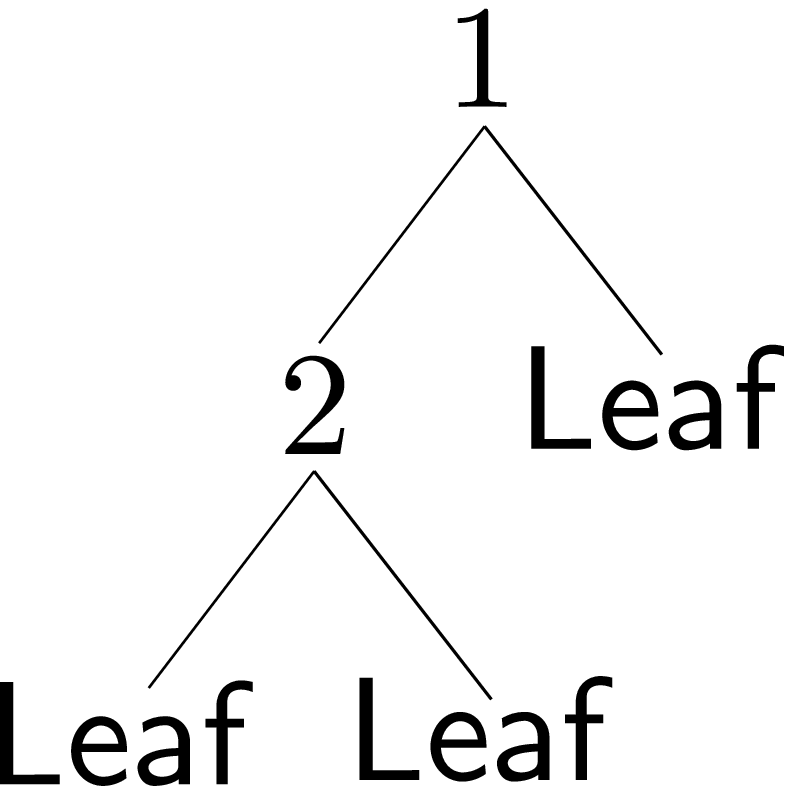}~~&~~\includegraphics[scale=0.15]{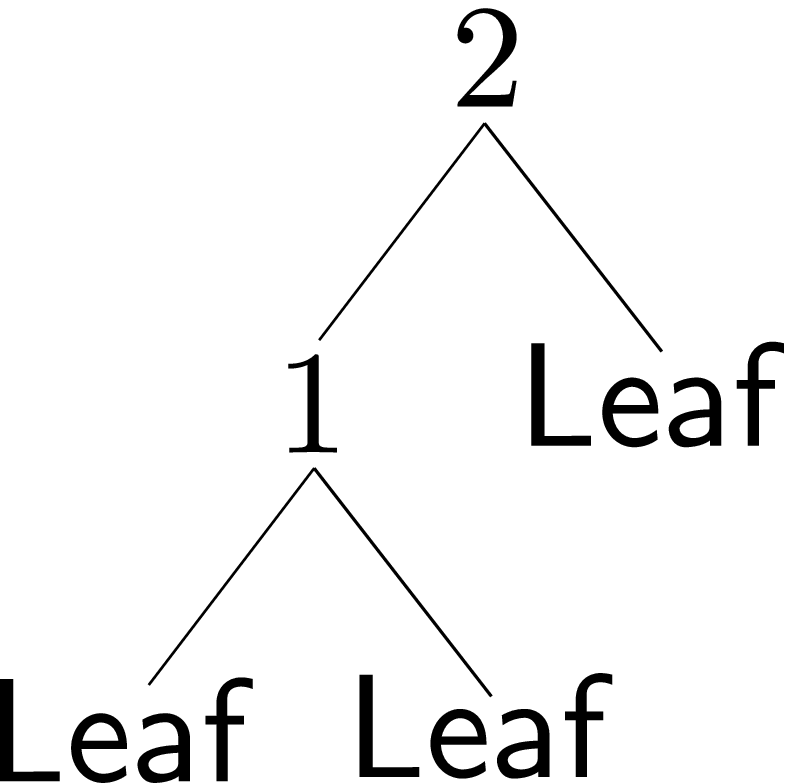}~~&~~\includegraphics[scale=0.15]{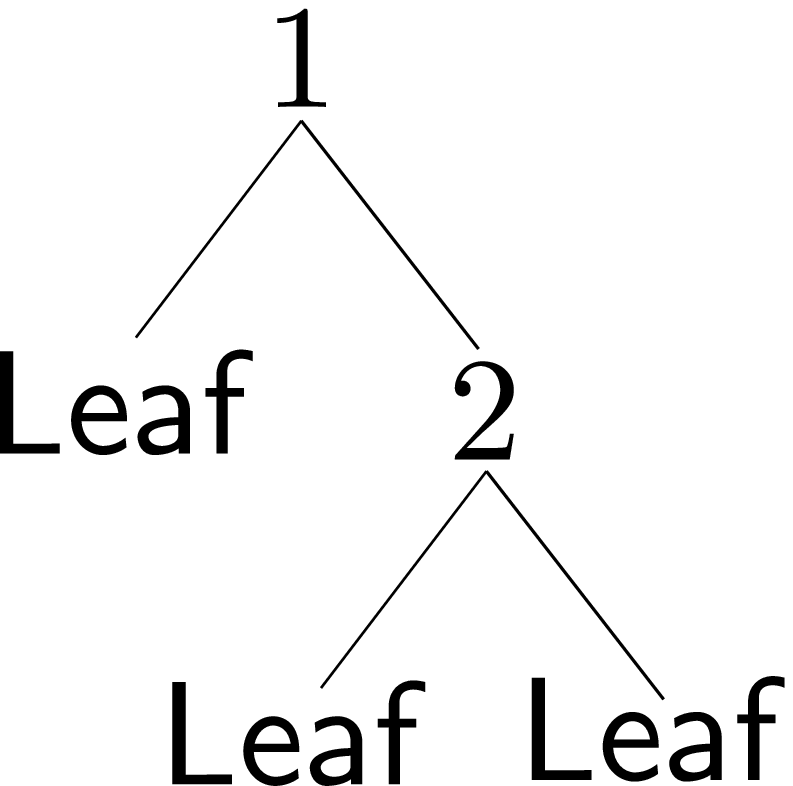}~~&~~\includegraphics[scale=0.15]{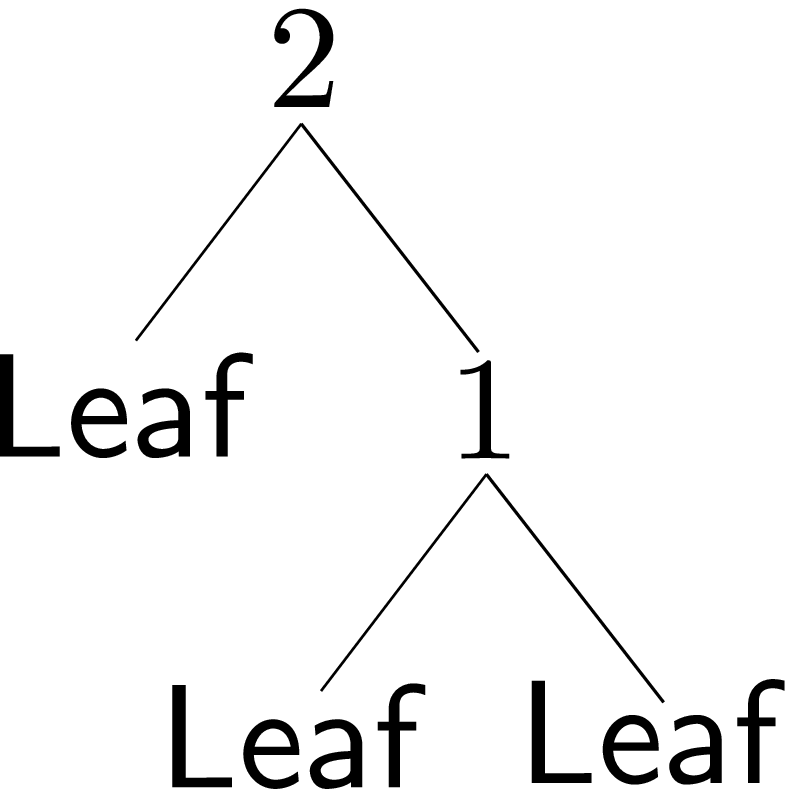}~~\\
\midrule
$\alpha(t)$ & $\{1\}$ & $\{1,2\}$ & $\{1,2\}$& $\{1,2\}$ & $\{1,2\}$\\
$\beta(t)$ & $1$ & $4$ & $4$& $4$ & $4$\\
\bottomrule
\end{tabular}
\label{table:example_beta}
\end{table}

Table \ref{table:example_beta} shows some examples of $\beta(t)$ with the \emph{Multiset} catamorphism.
The only tree that can map to $\{1\}$ by catamorphism \emph{Multiset} is $\Node(\Leaf, 1, \Leaf)$.
The last four trees are all the trees that can map to the multiset $\{1, 2\}$.
\exampleEndMark
\end{example}

Throughout this section we assume $\alpha : \rdttype \to \collection$
is a catamorphism defined by \kw{empty} and \kw{combine} with $\beta$
as defined in Definition~\ref{definition:beta}. We will prove that our
decision procedure is complete if $\alpha$ satisfies the {\em
  generalized sufficient surjectivity} condition defined in
Definition~\ref{definition:generalized_sufficient_surjectivity}.

\begin{definition}[Generalized sufficient surjectivity]
\label{definition:generalized_sufficient_surjectivity}
Given a catamorphism $\alpha$ and the corresponding function $\beta$
from Definition~\ref{definition:beta}, we say that $\alpha$ is a
generalized sufficiently surjective (GSS) catamorphism if it satisfies
the following condition. For every number $p \in \naturalnums$,
there exists some height $h_p \in \naturalnums$, computable as a
function of $p$, such that for every tree $t$ with $\height(t) \geq h_p$
we have $\beta(t) > p$.
\end{definition}

\begin{corollary}
\label{corollary:ss_is_gss}
Sufficiently surjective catamorphisms are GSS.
\end{corollary}
\begin{proof}
Let $\alpha$ be a sufficiently surjective catamorphism. By
Definition~\ref{definition:sufficient_surjectivity}, there exists a
finite set of shapes $S_p$ such that for every tree $t$, $\shape(t) \in
S_p$ or $\beta(t) > p$. Taking $h_p = 1 + \max\{ \height(t) ~|~ t
\in S_p\}$ ensures that for all trees $t$, if $\height(t) \geq h_p$
then $\beta(t) > p$, as needed. \qed
\end{proof}

We claim that our unrolling-based decision procedure with GSS
catamorphisms is (1) sound for proofs: if the procedure returns UNSAT,
then there are no models, (2) sound for models: every model returned
by the procedure makes the formula true, (3) terminating for
satisfiable formulas, and (4) terminating for unsatisfiable formulas.
We do not present the proofs for the first three properties, which can
be adapted with minor changes from similar proofs in
\cite{Suter2011SMR}. 
Rather, we focus on proving that Algorithm \ref{alg:revised_dp} is
terminating for unsatisfiable formulas.

In order to reason about our unrolling-based decision procedure we
define a related mathematical notion of unrolling called $n$-level
approximation. We show that for large enough values of $n$ the
$n$-level approximation of $\alpha$ can be used in place of $\alpha$
while preserving key satisfiability constraints. Finally, we show that
our unrolling-based decision procedure correctly uses uninterpreted
functions to model $n$-level approximations of $\alpha$.

\begin{definition}[$n$-level approximation]
\label{definition:n-level-approximation}
We say that $\alpha_0$ is a 0-level approximation of $\alpha$ iff
$\forall t\in\rdttype : \alpha_0(t) \in \range(\alpha)$. For $n > 0$ we
say $\alpha_n$ is an $n$-level approximation of $\alpha$ iff
\begin{align*}
\alpha_n(\Leaf) &= \kw{empty} \\
\alpha_n(\Node(t_L, e, t_R)) &= \kw{combine}(\alpha_{n-1}(t_L), e, \alpha_{n-1}(t_R))
\end{align*}
where $\alpha_{n-1}$ is an $(n-1)$-level approximation of $\alpha$.
\end{definition}

\begin{lemma}
\label{lemma:weaken-approximation}
If $\alpha_n$ is an $n$-level approximation of $\alpha$ then it is also
$m$-level approximation for all $0 \leq m < n$.
\end{lemma}
\begin{proof}
When $n = 0$ the result is vacuously true. For $n > 0$, induction on
$n$ shows that $\alpha_n$ is an $(n-1)$-level approximation
of $\alpha$. The result then follows by induction on $n - m$.
\qed
\end{proof}

Intuitively, an $n$-level approximation $\alpha_n$ of $\alpha$ always
returns values in the range of $\alpha$. In particular, $\alpha_n$
agrees with $\alpha$ on short trees, and on taller trees $\alpha_n$
provides values that $\alpha$ also provides for taller trees. These
intuitions are formalized in the next series of lemmas.

\begin{lemma}
\label{lemma:approximation-range}
Let $\alpha_n$ be an $n$-level approximation of $\alpha$. Then
$\range(\alpha_n) \subseteq \range(\alpha)$. Equivalently, for every $t$
there exists $t'$ such that $\alpha_n(t) = \alpha(t')$.
\end{lemma}
\begin{proof}
Straightforward induction on $n$.
\qed
\end{proof}

\begin{lemma}
\label{lemma:exact-approximation}
Let $\alpha_n$ be an $n$-level approximation of $\alpha$. If $\height(t)
< n$ then $\alpha_n(t) = \alpha(t)$.
\end{lemma}
\begin{proof}
Straightforward induction on $n$.
\qed
\end{proof}

\begin{lemma}
\label{lemma:approximation-replace-argument}
Let $\alpha_n$ be an $n$-level approximation of $\alpha$. If
$\height(t) \geq n$ then there exists $t'$
such that $\alpha_n(t) = \alpha(t')$ and $\height(t') \geq n$.
\end{lemma}
\begin{proof}
Induction on $n$. When $n = 0$ the result follows from the definition
of $\alpha_0$ and $\range(\alpha)$. In the inductive step let $t$ be a
tree such that $\height(t) \geq n$. Since $n > 0$ we have $t =
\Node(t_L, e, t_R)$ for some $t_L$, $e$, and $t_R$. By the definition
of $\alpha_n$ we have $\alpha_n(t) = \kw{combine}(\alpha_{n-1}(t_L),
e, \alpha_{n-1}(t_R))$ where $\alpha_{n-1}$ is an $(n-1)$-level
approximation of $\alpha$. By the definition of height we have
$\height(t_L) \geq n - 1$ or $\height(t_R) \geq n - 1$. Without loss of
generality, let us assume $\height(t_L) \geq n - 1$ (the argument for
$t_R$ is symmetric). Then by the inductive hypothesis there exists
$t_L'$ with $\alpha_{n-1}(t_L) = \alpha(t_L')$ and $\height(t_L') \geq
n - 1$. By Lemma~\ref{lemma:approximation-range} there exists $t_R'$
with $\alpha_{n-1}(t_R) = \alpha(t_R')$. Letting $t' = \Node(t_L', e,
t_R')$ we have
\begin{equation*}
\alpha_n(t) = \kw{combine}(\alpha_{n-1}(t_L), e, \alpha_{n-1}(t_R))
= \kw{combine}(\alpha(t_L'), e, \alpha(t_R'))
= \alpha(t')
\end{equation*}
Also, since $\height(t_L') \geq n - 1$ we have $\height(t') \geq n$.
\qed
\end{proof}

\begin{definition}
\label{definition:nice-form}
We say a formula $\phi$ in the parametric logic is in standard form if
it has the form $\phi_t \land \phi_{ce}$ where $\phi_t$ is a
conjunction of disequalities between distinct tree variables and
$\phi_{ce}$ is a formula in the $\collection\elem$ theory where
$\alpha$ is applied only to tree variables.
\end{definition}

In the following lemma we write $\phi[\alpha_n/\alpha]$ to denote the
formula $\phi$ with all occurrences of $\alpha$ replaced with $\alpha_n$.

\begin{lemma}
\label{lemma:sat-approx-implies-sat-nice-form}
Let $\phi$ be a formula in the standard form and let $\alpha$ be a GSS
catamorphism. Then there exists an $n$ such that if
$\phi[\alpha_n/\alpha]$ is satisfiable for some $n$-level
approximation $\alpha_n$ of $\alpha$ then $\phi$ is satisfiable.
\end{lemma}
\begin{proof}
Let $\phi$ have the form $\phi_t \land \phi_{ce}$ from
Definition~\ref{definition:nice-form}. Let $p$ be the number of
disequalities in $\phi_t$. By
Definition~\ref{definition:generalized_sufficient_surjectivity} there
exists a height $h_p$ such that for any tree $t$ of height greater
than or equal to $h_p$ we have $\beta(t) > p$. Let $n = h_p$.

Let $\alpha_n$ be an $n$-level approximation of $\alpha$ such that
$\phi[\alpha_n/\alpha]$ is satisfiable. Let $\model$ be a model of
$\phi[\alpha_n/\alpha]$. We will construct $\model'$, a modified
version of $\model$ with different values for the tree variables, such
that $\model'$ satisfies $\phi$. In particular, for each tree variable
$x$ in $\phi$ we will construct $\model'$ such that
$\alpha_n(\model(x)) = \alpha(\model'(x))$ and such that $\model'$ is
a model for the disequalities in $\phi_t$. This will ensure that
$\model'$ is a model of $\phi_{ce}$ since $\alpha$ is only applied to
tree variables in $\phi_{ce}$, and thus $\model'$ will be a model for
$\phi$. We construct $\model'$ by considering each tree variable in
turn.

Let $x$ be a tree variable in $\phi$. Let $T = \model(x)$ be the
concrete value of $x$ in the model $\model$. If $\height(T) <
n$ then $\alpha_n(T) = \alpha(T)$ by
Lemma~\ref{lemma:exact-approximation}. In that case we can take
$\model'(x) = \model(x)$. Otherwise, $\height(T) \geq n$ and
by Lemma~\ref{lemma:approximation-replace-argument} there exists a
$T'$ such that $\height(T') \geq n$ and $\alpha_n(T) =
\alpha(T')$. By
Definition~\ref{definition:generalized_sufficient_surjectivity} we
have $\beta(T') > p$. That is, we have more than $p$ distinct trees
$T''$ such that $\alpha(T'') = \alpha(T') = \alpha_n(T)$. Since we
have $p$ disequalities in $\phi_t$ there are at most $p$ forbidden
values for $T''$ in order to satisfy $\phi_t$. Thus we can always make
a selection for $x$ in $\model'$ such that the $p$ disequalities are
still satisfied. \qed
\end{proof}

\begin{lemma}
\label{lemma:translate-nice-form}
A formula $\phi$ in the parametric logic can be translated to an
equisatisfiable formula $\phi_1 \lor \cdots \lor \phi_k$ such that
each $\phi_i$ is in standard form.
\end{lemma}
\begin{proof}
We prove this lemma by providing a series of translation steps from
$\phi$ to a disjunction of formulas in standard form. Each step of the
translation will produce an equisatisfiable formula which is closer
to standard form.

To simplify the presentation of the translation steps, we write
$\phi[e]$ to indicate that $e$ is an expression that appears in
$\phi$, and we then write $\phi[e']$ to denote $\phi$ with all
occurrences of $e$ replaced by $e'$.

Many of these translation steps will introduce new variables. We will
write such variables as $\bar{x}$, with a line over them, to emphasizes
they they are fresh variables.

\paragraph{Step 1 (DNF)} Convert $\phi$ to disjunctive normal form. It then
suffices to show that each conjunctive clause can be converted into a
disjunction of standard form formulas.

\paragraph{Step 2 (Eliminate Selectors)} Given a conjunctive clause
$\phi$ we eliminate all selectors $\kw{left}(t)$, $\kw{elem}(t)$, and
$\kw{right}(t)$ by repeatedly applying the following conversions.
\begin{align*}
\phi[\kw{left}(t)] &~~\leadsto~~ t = \Node(\bar{t}_L, \bar{e}, \bar{t}_R)
\land \phi[\bar{t}_L] \\[3pt]
\phi[\kw{elem}(t)] &~~\leadsto~~ t = \Node(\bar{t}_L, \bar{e}, \bar{t}_R)
\land \phi[\bar{e}] \\[3pt]
\phi[\kw{right}(t)] &~~\leadsto~~ t = \Node(\bar{t}_L, \bar{e}, \bar{t}_R)
\land \phi[\bar{t}_R]
\end{align*}
This results in a conjunctive clause with no selectors.

\paragraph{Step 3 (Tree Unification)} Given a conjunctive clause $\phi$
with no tree selectors, we now eliminate all equalities between tree
terms. Such tree equalities can only appear as top level conjuncts in
the clause. Let $\phi = \phi_{eq} \land \phi'$ where $\phi_{eq}$
contains all of the tree equalities. We eliminate the equalities by
doing first-order term unification with the modification that
constraints between terms in the element theory are left unsolved. If
unification fails then we can replace $\phi$ with $\bot$ since the
clause is unsatisfiable. If unification succeeds it returns a most
general unifier $\sigma$ and a conjunction of element theory
equalities $E$. Then this step produces $\phi'\sigma \land E$ which is
a conjunctive clause with no tree selectors or tree equalities.

\paragraph{Step 4 (Reduce Disequalities)} Given a conjunctive clause
with no tree selectors or tree equalities, we now reduce all tree term
disequalities to disequalities between distinct tree variables. We do
this by repeatedly applying the following transformations. In these
translations $t_v$ stands for a tree variable. To save space, we treat
$t_1 \neq t_2$ and $t_2 \neq t_1$ as equivalent for these
translations.
\begin{align*}
\Leaf \neq \Leaf \land \phi &~~\leadsto~~ \bot \\[3pt]
\Leaf \neq \Node(t_L, e, t_R) \land \phi &~~\leadsto~~ \phi \\[3pt]
\Node(t_L, e, t_R) \neq \Node(t_L', e', t_R') \land \phi &~~\leadsto~~ ((t_L
\neq t_L') \land \phi) \lor~ \\
&\phantom{~~\leadsto~~} ((e \neq e') \land \phi) \lor~ \\
&\phantom{~~\leadsto~~} ((t_R \neq t_R') \land \phi) \\[3pt]
t_v \neq \Leaf \land \phi[t_v] &~~\leadsto~~ \phi[\Node(\bar{t}_L, \bar{e},
  \bar{t}_R)] \\[3pt]
t_v \neq \Node(t_L', e', t_R') \land \phi[t_v] &~~\leadsto~~ \phi[\Leaf]
\lor~ \\
&\phantom{~~\leadsto~~} (\bar{t}_L \neq t_L' \land
\phi[\Node(\bar{t}_L, \bar{e}, \bar{t}_R)]) \lor~ \\
&\phantom{~~\leadsto~~} (\bar{e} \neq e' \land \phi[\Node(\bar{t}_L, \bar{e},
  \bar{t}_R)]) \lor~ \\
&\phantom{~~\leadsto~~} (\bar{t}_R \neq t_R' \land \phi[\Node(\bar{t}_L, \bar{e},
  \bar{t}_R)]) \\
t_v \neq t_v \land \phi &~~\leadsto~~ \bot
\end{align*}
The termination of these transformations is obvious since the term
size of the leading disequality is always getting smaller. Note that
some transformations may produce a disjunction of conjunctive clauses.
This is not a problem since the initial conjunctive clause that we
focus on is already part of a top-level disjunction.

After these transformations we will have a disjunction of conjunctive
clauses where each clause has no tree selectors or tree equalities,
and all tree disequalities are between distinct tree variables.

\paragraph{Step 5 (Partial Evaluation of $\alpha$)} Given a conjunctive
clause $\phi$ where each clause has no tree selectors or tree
equalities, and all tree disequalities are between distinct tree
variables, we now partially evaluate $\alpha$. We do this by
repeatedly applying the following transformations.
\begin{align*}
\phi[\alpha(\Leaf)] &~~\leadsto~~ \phi[\kw{empty}] \\[3pt]
\phi[\alpha(\Node(t_L, e, t_R))] &~~\leadsto~~
\phi[\kw{combine}(\alpha(t_L), e, \alpha(t_R))]
\end{align*}
After these transformations we will have a conjunctive clause where
there are no tree selectors or tree equalities, all tree disequalities
are between distinct tree variables, and $\alpha$ is applied only to
tree variables. Thus the clause is in standard form. Therefore the
original formula will be transformed into a disjunction of standard
form clauses. \qed
\end{proof}

\begin{lemma}
\label{lemma:sat-approx-implies-sat}
Given a formula $\phi$ in the parametric logic and a GSS catamorphism
$\alpha$, there exists an $n$ such that if $\phi[\alpha_n/\alpha]$ is
satisfiable for some $n$-level approximation $\alpha_n$ of $\alpha$
then $\phi$ is satisfiable.
\end{lemma}
\begin{proof}
By Lemma~\ref{lemma:translate-nice-form} we can translate $\phi$ to
the equisatisfiable formula $\phi_1 \lor \cdots \lor \phi_k$ where
each $\phi_i$ is in standard form. By
Lemma~\ref{lemma:sat-approx-implies-sat-nice-form}, for each $\phi_i$
we have an $n_i$ such that if $\phi_i[\alpha_{n_i}/\alpha]$ is
satisfiable for some $n_i$-level approximation $\alpha_{n_i}$ of
$\alpha$ then $\phi_i$ is satisfiable. Let $n = \max\{n_1, \ldots,
n_k\}$. By Lemma~\ref{lemma:weaken-approximation} we have for each
$\phi_i$ that if $\phi_i[\alpha_{n}/\alpha]$ is satisfiable for some
$n$-level approximation $\alpha_{n}$ of $\alpha$ then $\phi_i$ is
satisfiable. Thus if $\phi[\alpha_{n}/\alpha]$ is satisfiable for some
$n$-level approximation $\alpha_{n}$ of $\alpha$ then $\phi$ is
satisfiable. \qed
\end{proof}

\begin{theorem}
\label{theorem:unsat-implies-unsat-approx}
Given a formula $\phi$ in the parametric logic and a $GSS$
catamorphism $\alpha$, there exists an $n$ such that if $\phi$ is
unsatisfiable then $\phi[\alpha_n/\alpha]$ is unsatisfiable for all
$n$-level approximations $\alpha_n$ of $\alpha$.
\end{theorem}
\begin{proof}
This is the contrapositive of
Lemma~\ref{lemma:sat-approx-implies-sat}.
\qed
\end{proof}

\begin{theorem}
Given an unsatisfiable formula $\phi$ in the parametric logic, a $GSS$
catamorphism $\alpha$, and a correct range predicate $R_\alpha$,
Algorithm~\ref{alg:revised_dp} is terminating.
\end{theorem}
\begin{proof}
Let $\phi$ be an unsatisfiable formula. By
Theorem~\ref{theorem:unsat-implies-unsat-approx}, there exists an $n$
such that $\phi[\alpha_n/\alpha]$ is unsatisfiable for all $n$-level
approximations $\alpha_n$ of $\alpha$.

In Algorithm~\ref{alg:revised_dp}, $\alpha$ is initially replaced by
an uninterpreted function $\uf$ which is unrolled during the
algorithm. Consider the $n^\text{th}$ unrolling together with the
range restrictions and call the resulting formula $\phi_n =
\phi[\uf/\alpha] \land C_n$. Note that $C_0$ is the initial range
constraints on $\uf$ without any unrolling. We will show that $\phi_n$
is unsatisfiable and thus the algorithm will terminate with UNSAT
within the first $n$ unrollings.

Suppose, towards contradiction, that $\phi_n$ is satisfiable. Let
$\model$ be a model of $\phi_n$. It is not necessarily true that
$\model(\uf)$ is an $n$-level approximation of $\alpha$ since it may,
for example, return any value for inputs to which $\uf$ is not applied
in $\phi_n$. However, for the values to which $\uf$ is applied in
$\phi[\uf/\alpha]$ it acts like an $n$-level approximation of $\alpha$
due to the constraints imposed by $C_n$ and by the correctness of
$R_\alpha$. Thus, we can construct a new model $\model'$ which differs
from $\model$ only in the value of $\uf$ and such that: (1)
$\model(\uf)$ and $\model'(\uf)$ agree on all values to which $\uf$ is
applied in $\phi[\uf/\alpha]$ and (2) $\model'(\uf)$ is an $n$-level
approximation of $\alpha$.

By construction, $\model'$ satisfies $\phi[\uf/\alpha]$. Therefore
$\model'$ satisfies $\phi[\model'(\uf)/\alpha]$ which contradicts the
fact that $\phi[\alpha_n/\alpha]$ is unsatisfiable for all $n$-level
approximations $\alpha_n$ of $\alpha$. Thus $\phi_n$ must be
unsatisfiable. \qed
\end{proof}

This proof implies that Algorithm~\ref{alg:revised_dp} terminates for
unsatisfiable formulas after a bounded number of unrollings based on
$\phi$ and $\alpha$. If we compute this bound explicitly, then we can
terminate the algorithm early with SAT after the computed number of
unrollings. However, if we are interested in complete tree models (in which all variables are assigned concrete values), we
still need to keep unrolling until we reach line
\ref{reviseddpcode:sat_case} in Algorithm \ref{alg:revised_dp}.

The bound on the number of unrollings needed to check unsatisfiability
depends on two factors. First, the structure of $\phi$ gives rise to a
number of tree variable disequalities in the conversion of
Lemma~\ref{lemma:translate-nice-form}. Second, the unrolling bound
depends on the relationship between $p$ and $h_p$ in
Definition~\ref{definition:generalized_sufficient_surjectivity}. In
Section \ref{section:monotonic_cata}, we show that the unrolling bound
is linear (Theorem \ref{theorem:monotonic_generalized_sufficient_surjectivity}) in the number of disequalities for a class of catamorphisms
called \emph{monotonic} catamorphisms. Later, in
Section~\ref{section:assoc_cata}, we show that this bound can be made
logarithmic (Lemma \ref{lemma:minbeta_catalan}) in the number of disequalities for a special, but common,
form of catamorphisms called \emph{associative} catamorphisms.

In practice, computing the exact bound on the number of unrollings is
impractical. The conversion process described in
Lemma~\ref{lemma:translate-nice-form} is focused on correctness rather
than efficiency. Instead, it is much more efficient to do only the
unrolling of $\alpha$ and leave all other formula manipulation to an
underlying SMT solver. Even from this perspective, the bounds we
establish on the number of unrollings are still useful to explain why
the procedure is so efficient is practice.



\section{Monotonic Catamorphisms}
\label{section:monotonic_cata}
We now propose {\em monotonic} catamorphisms and prove that Algorithm \ref{alg:revised_dp} is complete for this class by showing that monotonic catamorphisms satisfy the GSS condition.  We also show that this class is a subset of sufficiently surjective catamorphisms, but general enough to include all catamorphisms described in~\cite{Suter2010DPA,Suter2011SMR} and all those that we have run into in industrial experience.  Monotonic catamorphisms admit a termination argument in terms of the number of unrollings, which is an open problem in \cite{Suter2011SMR}.  Moreover, a subclass of monotonic catamorphisms, {\em associative} catamorphisms can be combined while preserving completeness of the formula, addressing another open problem in~\cite{Suter2011SMR}.

\subsection{Monotonic Catamorphisms}

A catamorphism $\alpha$ is \emph{monotonic} if for every ``high
enough'' tree $t \in \rdttype$, either $\beta(t) = \infty$ or there
exists a tree $t_0 \in \rdttype$ such that $t_0$ is smaller than $t$
and $\beta(t_0) < \beta(t)$. Intuitively, this condition ensures that
the more number of unrollings we have, the more candidates SMT solvers
can assign to tree variables to satisfy all the constraints involving
catamorphisms. Eventually, the number of tree candidates will be large
enough to satisfy all the constraints involving tree equalities and
disequalities among tree terms, leading to the completeness of the
procedure.

\begin{definition}[Monotonic catamorphisms]
\label{definition:mono_cata}
A catamorphism $\alpha: \rdttype \rightarrow \collection$ is monotonic iff
there exists a constant $\halpha \in \naturalnums$ such that:
\begin{align*}
\forall t \in \rdttype: \height(t) \geq \halpha \Rightarrow \bigl ( & \beta(t) = \infty~\vee\\
& \exists t_0 \in \rdttype: \height(t_0) = \height(t) - 1 \wedge \beta(t_0) < \beta(t)\bigr )
\end{align*}
\end{definition}

Note that if $\alpha$ is monotonic with $\halpha$, it is also
monotonic with any $\halpha' \geq h_\alpha$.
In our decision procedure, we assume that if $\alpha$ is monotonic, the range of $\alpha$ can be expressed precisely as a predicate $R_\alpha$.

\begin{example}[Monotonic catamorphisms]
Catamorphism $\DW$ in Section \ref{section:dp_by_example} is monotonic with $\halpha = 1$ and $\emph{Multiset}$ is monotonic with $\halpha = 2$.
An example of a non-monotonic catamorphism is \Mirror~in \cite{Suter2010DPA}:
\[
  \Mirror(t) =
    \begin{cases}
    \Leaf  &\text{if $t = \Leaf$}\\
    \Node\bigl(\Mirror(t_R), e, \Mirror(t_L)\bigr) &\text{if $t = \Node(t_L, e, t_R)$}\\
    \end{cases}
\]
Because $\forall t \in \rdttype: \beta_{\Mirror}(t) = 1$, the catamorphism is not monotonic. We will discuss in detail some examples of monotonic catamorphisms in Section \ref{subsection:examples_mono_catas}.
\exampleEndMark
\end{example}

\subsection{Some Properties of Monotonic Catamorphisms}
\label{subsection:predicate_satisfying_suff_surj}

\begin{definition}[$\minbeta$]
\label{definition:minbeta}
$\minbeta(h)$ is the minimum of $\beta(t)$ of all trees $t$ of height $h$:
\[ \minbeta(h) = \min\{\beta(t)~|~t \in \rdttype, \height(t) = h\}\]
\end{definition}

\begin{corollary}
\label{corollary:positive_min_beta}
$\minbeta(h)$ is always greater or equal to 1.
\end{corollary}
\begin{proof}
$\minbeta(h) \geq 1$ since $\forall t \in \rdttype: \beta(t) = |\alpha^{-1}\bigl(\alpha(t)\bigr)|\geq 1$.
\qed
\end{proof}

\begin{lemma} [Monotonic Property of $\minbeta$]
\label{lemma:minbeta_monotonic}
Function $\minbeta: \naturalnums \rightarrow \naturalnums \cup
\{\infty\}$ satisfies the following monotonic property:
\[ \begin{array}{lll}
\forall h \in \naturalnums, h \geq \halpha: & \minbeta(h) = \infty \Rightarrow \minbeta(h + 1) = \infty&\wedge \\
& \minbeta(h) < \infty \Rightarrow \minbeta(h) < \minbeta(h + 1)&
\end{array}\]
\end{lemma}
\begin{proof}

Consider any $h \in \naturalnums$ such that $h \geq \halpha$. There are two cases to consider: $\minbeta(h) = \infty$ and $\minbeta(h) < \infty$.

\underline{Case 1} [$\minbeta(h) = \infty$]:
From Definition \ref{definition:minbeta}, every tree $t_h$ of height $h$ has $\beta(t_h) = \infty$ because $\minbeta(h) = \infty$.
Hence, every tree $t_{h + 1}$ of height $h + 1$ has $\beta(t_{h + 1}) = \infty$ from Definition \ref{definition:mono_cata}.
Thus, $\minbeta(h + 1) = \infty$.

\underline{Case 2} [$\minbeta(h) < \infty$]: Let $t_{h + 1}$ be any tree of height $h + 1$.
From Definition \ref{definition:mono_cata},
there are two sub-cases as follows.
\begin{itemize}
\item \underline{Sub-case 1} [$\beta(t_{h + 1}) = \infty$]: Because $\minbeta(h) < \infty$, we have $\minbeta(h) < \beta(t_{h + 1})$.
\item \underline{Sub-case 2} [there exists $t_{h}$ of height $h$ such that $\beta(t_h) < \beta(t_{h + 1})$]:
From Definition \ref{definition:minbeta}, $\minbeta(h) < \beta(t_{h + 1})$.
\end{itemize}

In both sub-cases, we have
$\minbeta(h) < \beta(t_{h + 1})$.
Since $t_{h + 1}$ can be any tree of height $h + 1$, we have
$\minbeta(h) < \minbeta(h + 1)$ from Definition \ref{definition:minbeta}.
\qed
\end{proof}

\begin{corollary} For any natural number $p > 0$, $\minbeta(\halpha + p) > p$.
\label{corollary:minbeta}
\end{corollary}
\begin{proof}
By induction on $p$ based on Lemma \ref{lemma:minbeta_monotonic} and Corollary \ref{corollary:positive_min_beta}.
\qed
\end{proof}

\begin{theorem}
\label{theorem:monotonic_generalized_sufficient_surjectivity}
Monotonic catamorphisms are GSS (Definition \ref{definition:generalized_sufficient_surjectivity}).
\end{theorem}
\begin{proof}
Let $\alpha$ be a monotonic catamorphism with $\halpha$ as in
Definition~\ref{definition:mono_cata}. Let $h_p = \halpha + p$. From
Corollary \ref{corollary:minbeta}, $\minbeta(h_p) > p$. Based on Lemma
\ref{lemma:minbeta_monotonic}, we can show by induction on $h$ that
$\forall h \geq h_p: \minbeta(h) > p$. By
Definition~\ref{definition:minbeta}, for every tree $t$ with
$\height(t) \geq h_p$ we have $\beta(t) > p$. Therefore $\alpha$ is
GSS. \qed
\end{proof}

The proof of Theorem \ref{theorem:monotonic_generalized_sufficient_surjectivity} shows that monotonic catamorphisms admit a linear bound on the number of unrollings needed to establish unsatisfiability in our procedure.

\subsection{Examples of Monotonic Catamorphisms}
\label{subsection:examples_mono_catas}
This section proves that all sufficiently surjective catamorphisms introduced by Suter et al.~\cite{Suter2010DPA} are monotonic.
These catamorphisms are listed in Table \ref{table:POPL_catas}.
Note that the \emph{Sortedness} catamorphism can be defined to allow or not allow duplicate elements~\cite{Suter2010DPA};
we define $\emph{Sortedness}_\emph{dup}$ and $\emph{Sortedness}_\emph{nodup}$ for the \emph{Sortedness} catamorphism
where duplications are allowed and disallowed, respectively.

The monotonicity of \emph{Set}, \emph{SizeI}, \emph{Height}, \emph{Some}, \emph{Min}, and \emph{Sortedness}$_\emph{dup}$ catamorphisms
is easily proved by showing the relationship between infinitely surjective abstractions (see Definition \ref{definition:infinitely_surjective_catamorphisms})
and monotonic catamorphisms.

\begin{lemma} Infinitely surjective abstractions are monotonic.
\label{lemma:infinitely_surjective_abstraction_monotonic}
\end{lemma}
\begin{proof}
According to Definition \ref{definition:infinitely_surjective_catamorphisms},
$\alpha$ is infinitely surjective $S$-abstraction, where $S$ is a set of trees, if and only if
$\beta(t)$ is finite for $t \in S$ and infinite for $t \notin S$.
Therefore, $\alpha$ is monotonic with $\halpha = \max\{\height(t)~|~t \in S\} + 1$.
\qed
\end{proof}

\begin{theorem} \label{theorem:set} Set, SizeI, Height, Some, Min, and Sortedness$_\mathit{dup}$ are monotonic.
\end{theorem}
\begin{proof}
\cite{Suter2010DPA} showed that \emph{Set}, \emph{SizeI}, \emph{Height}, and $\emph{Sortedness}_\emph{dup}$
are infinitely surjective abstractions.
Also, \emph{Some} and \emph{Min} have the properties of infinitely surjective \{\Leaf\}-abstractions.
Therefore, the theorem follows from Lemma \ref{lemma:infinitely_surjective_abstraction_monotonic}.
\qed
\end{proof}

It is more challenging to prove that \emph{Multiset}, \emph{List}, and
$\emph{Sortedness}_\emph{nodup}$ catamorphisms are monotonic since they are
not infinitely surjective abstractions. In
Section~\ref{section:assoc_cata} we will introduce the notion of
associative catamorphisms which includes \emph{Multiset}, $\emph{List}_\emph{inorder}$, and $\emph{Sortedness}_\emph{nodup}$, and prove in
Theorem~\ref{theorem:assoc_cata_are_monotonic} that all associative
catamorphisms are monotonic. For now, we conclude this section by
showing that the all \emph{List} catamorphisms are monotonic.

\begin{theorem}
\label{theorem:list}
List catamorphisms are monotonic.
\end{theorem}
\begin{proof}
Let $\alpha$ be a \emph{List} catamorphism.
For any tree $t$ there are exactly $\numshapes\bigl(\size(t)\bigr)$ distinct trees that can map to $\alpha(t)$.
 This is true because: (1) the length of the list
$\alpha(t)$ is the number of internal nodes of $t$, (2) there are
exactly $\numshapes\bigl(\size(t)\bigr)$ tree shapes with the same
number of internal nodes as $t$, and (3) the order of elements in
$\alpha(t)$ gives rise to exactly one tree with the same catamorphism
value $\alpha(t)$ for each tree shape. Thus,
$\beta(t) = \numshapes\bigl(\size(t)\bigr)$.

Let $\halpha = 2$ and let $t$ be an arbitrary tree with $\height(t)
\geq 2$. Then there exists $t_0$ such that $\height(t_0) = \height(t) -
1$ and $\size(t_0) < \size(t)$. By Property \ref{property:size_height},
$\height(t) \geq \halpha = 2$ implies $\size(t) \geq 5$. By Lemma
\ref{lemma:numshapes}, $\numshapes\bigl(\size(t_0)\bigr) <
\numshapes\bigl(\size(t)\bigr)$, which means $\beta(t_0) < \beta(t)$.
Therefore $\alpha$ is monotonic. \qed
\end{proof}


\subsection{Monotonic Catamorphisms are Sufficiently Surjective}
In this section, we demonstrate that monotonic catamorphisms are a strict subset of sufficiently surjective catamorphisms.  To this end, we prove that all monotonic catamorphisms are sufficiently surjective, and then provide a witness catamorphism to show that there exists a sufficiently surjective catamorphism that is not monotonic.  Although this indicates that monotonic catamorphisms are less general, the constructed catamorphism is somewhat esoteric and we have not found any practical application in which a catamorphism is sufficiently surjective but not monotonic.

To demonstrate that monotonic catamorphisms are sufficiently
surjective, we describe a predicate $M_h$ (for each $h$) that is
generic for any monotonic catamorphism. We show that $M_h(\alpha(t))$
holds for any tree with $\height(t) \geq h$ in
Lemma~\ref{lemma:image_of_t}. We then show that if $M_h(c)$ holds then
there exists a tree $t$ with $\height(t) \geq h$ and $\alpha(t) = c$ in
Lemma~\ref{lemma:mp_height_lemma}. The proof of sufficient
surjectivity follows directly from these lemmas.

\begin{definition}[$M_h$ for arbitrary catamorphism $\alpha$]
Let $\alpha$ be defined by $\kw{empty}$ and $\kw{combine}$. Let
$R_\alpha$ be a predicate which recognizes exactly the range of
$\alpha$. We define $M_h$ recursively as follows:
\begin{align*}
M_0(c) &= R_\alpha(c) \\
M_{h+1}(c) &= \exists c_L, e, c_R : c = \kw{combine}(c_L, e, c_R)
\land~ \\
& \phantom{~= \exists c_L, e, c_R :}
((M_h(c_L) \land R_\alpha(c_R)) \lor (R_\alpha(c_L) \land M_h(c_R))
\end{align*}
Note that $M_h(c)$ may have nested existential quantifiers, but these
can always be moved out to the top-level.
\end{definition}

\begin{lemma}
\label{lemma:image_of_t}
If $\height(t) \geq h$ then $M_h(\alpha(t))$ holds.
\end{lemma}
\begin{proof}
Induction on $h$. In the base case, $h = 0$. Given an arbitrary tree
$t$, we have $M_0(\alpha(t)) = R_\alpha(\alpha(t))$ which holds by the
definition of $R_\alpha$.

In the inductive case assume the formula holds for a fixed $h$, and
let $t$ be an arbitrary tree with $\height(t) \geq h + 1$. Let $t =
\Node(t_L, e, t_R)$ where either $\height(t_L) = \height(t) - 1$ or
$\height(t_R) = \height(t) - 1$. Without loss of generality assume that
$\height(t_L) = \height(t) - 1$; the other case is symmetric. Then we
have $\height(t_L) \geq h$ and so by the inductive hypothesis
$M_h(\alpha(t_L))$ holds. We also know $R_\alpha(\alpha(t_R))$ by the
definition of $R_\alpha$. Therefore $M_{h+1}(\alpha(t))$ holds. \qed
\end{proof}

\begin{lemma}
\label{lemma:mp_height_lemma}
If $M_h(c)$ holds, there exists $t$ such that $\height(t)
\geq h$ and $c = \alpha(t)$.
\end{lemma}
\begin{proof}
Induction on $h$. In the base case, $h = 0$. We know $M_0(c)$ holds
and so $R_\alpha(c)$ holds. By the definition of $R_\alpha$ there
exists a tree $t$ such that $c = \alpha(t)$. Trivially, $\height(t)
\geq 0$.

In the inductive case assume the formula holds for a fixed $h$, and
assume $M_{h+1}(c)$ holds. Expanding the definition of $M_{h+1}$ we
know that there exists $c_L$, $e$, and $c_R$ such that $c =
\kw{combine}(c_L, e, c_R)$ and either $M_h(c_L) \land R_\alpha(c_R)$
or $R_\alpha(c_L) \land M_h(c_R)$ holds. Without loss of generality assume the
former; the latter case is symmetric. Given $M_h(c_L)$ holds the
inductive hypothesis gives us a tree $t_L$ with $\height(t_L) \geq h$
and $c_L = \alpha(t_L)$. Given $R_\alpha(c_R)$ we know there exists
$t_R$ with $c_R = \alpha(t_R)$ by the definition of $R_\alpha$. Let $t
= \Node(t_L, e, t_R)$. Since $\height(t_L) \geq h$ we have $\height(t)
\geq h + 1$. Moreover,
\begin{equation*}
\alpha(t) = \alpha(\Node(t_L, e, t_R)) = \kw{combine}(\alpha(t_L), e,
\alpha(t_R)) = \kw{combine}(c_L, e, c_R) = c
\end{equation*}
Thus $t$ has the required properties. \qed
\end{proof}

\begin{theorem}
\label{theorem:mono_catas_are_suff_surj}
Monotonic catamorphisms are sufficiently surjective.
\end{theorem}

\begin{proof}
Let $\alpha$ be a monotonic catamorphism, and let $h_p$ be as given in
Theorem~\ref{theorem:monotonic_generalized_sufficient_surjectivity}.
We define $S_{p} = \{\shape(t)~|~\height(t) < h_{p} \}$, and show that
for each 
tree $t$ either $\shape(t) \in S_{p}$ or that $M_{h_p}\bigl(\alpha(t)\bigr)$
holds for the tree. We partition trees by height. If a tree is shorter
in height than $h_p$ then it is captured by $S_p$. Otherwise, by
Lemma~\ref{lemma:image_of_t}, $M_{h_p}(\alpha(t))$ holds.

Now we need to satisfy the constraints on $S_{p}$ and $M_{h_p}$.
First, we note that the set $S_{p}$ is clearly finite. Second, we show
that $M_{h_p}(c)$ implies $|\alpha^{-1}(c)| > p$. If $M_{h_p}(c)$
holds, then by Lemma~\ref{lemma:mp_height_lemma} there exists a tree
$t$ of height greater than or equal to $h_{p}$ such that $c =
\alpha(t)$, and then as in
Theorem~\ref{theorem:monotonic_generalized_sufficient_surjectivity},
$|\alpha^{-1}(c)| > p$. \qed
\end{proof}

\newcommand{\almostid}{id^{*}}

We next demonstrate that there are sufficiently surjective
catamorphisms that are not monotonic. Consider the following ``almost
identity'' catamorphism $\almostid$ that maps a tree of unit elements (i.e., null)
to a pair of its height and another unit-element tree:
\begin{equation*}
   \almostid(t) =
     \begin{cases}
     \bigl\langle 0, \Leaf \bigr\rangle  &\text{if $t = \Leaf$}\\
     \bigl\langle h, tt \bigr\rangle &\text{if $t = \Node(t_L, (), t_R)$}
     \end{cases}
\end{equation*}
where:
\begin{align*}
 h &= 1 + \max\{\almostid(t_L).\emph{first}, \almostid(t_R).\emph{first}\}  \\
 tt&= \begin{cases}
     \Node(\almostid(t_L).\emph{second}, (), \almostid(t_R).\emph{second})  &\text{if $h$ is odd}\\
     \Node(\almostid(t_L).\emph{second}, (), \Leaf) &\text{if $h$ is even}
     \end{cases}
\end{align*}

If the height of the tree is odd, then it returns a tree with the same
top-level structure as the input tree. If the height is even, it
returns a tree whose left side is the result of the catamorphism and
whose right side is $\Leaf$.

\newcommand{\stcount}{count_{st}}
\newcommand{\subtrees}{st}

\begin{definition}
\label{definition:subtrees}
Let $\subtrees(h)$ be the set of all trees with unit element of height
less or equal to $h$. We can construct $\subtrees(h)$ as follows:
\begin{equation*}
   \subtrees(h) =
     \begin{cases}
        {\Leaf}  &\text{if $h = 0$}\\
        \{ \Node(l, (), r)~|~l,r \in \subtrees(h-1)\}~\cup~\{\Leaf\} &\text{if $h > 0$}
     \end{cases}
\end{equation*}
\end{definition}

\begin{definition}
\label{definition:stcount}
Let $\stcount(h) = |\subtrees(h)|$, i.e., the size of the set of all trees of unit element of
height less than or equal to $h$.
\end{definition}

\begin{corollary}
\begin{equation*}
   \stcount(h) =
     \begin{cases}
        1  &\text{if $h = 0$}\\
        \bigl(\stcount(h-1)\bigr)^2 + 1 &\text{if $h > 0$}
     \end{cases}
\end{equation*}
\end{corollary}
\begin{proof}
Induction on $h$. The base case is trivial. For the inductive case, we
have
\begin{align*}
\stcount(h+1) &= |\subtrees(h+1)| \\
&= \bigl|\{ \Node(l, (), r)~|~l,r \in \subtrees(h)\}\bigr| + 1 \\
&= |\subtrees(h)|^2 + 1 \\
&= \bigl(\stcount(h)\bigr)^2 + 1
\end{align*}
Thus the equation holds for all $h$. \qed
\end{proof}

Note that the number of trees grows quickly; e.g., the values for $h
\in 0..5$ are $1, 2, 5, 26, 677, 458330$.

\begin{theorem}
\label{thm:height_is_ss}
The $\almostid$ catamorphism is sufficiently surjective.
\end{theorem}

\begin{proof}
For a given natural number $p$, we choose $S_{p}$ and $M_{p}$ as follows:
\begin{itemize}
    \item $S_{p} = \{t~|~\height(t) \leq p + 5\}$, and
    \item $M_{p}(\langle h, t \rangle) = h > p + 5 \land R_{\alpha}(h, t)$
\end{itemize}

\noindent where:
\begin{equation*}
   R_{\alpha}(h, t) = (\height(t) = h \land \shape(t, h))
\end{equation*}
and:
\begin{equation*}
 \small
   \shape(t, k) =
     \begin{cases}
        \kw{true}  &\text{if $t = \Leaf$}\\
        \shape(t_L, k-1) \land \shape(t_R, k-1) & \text{if $k$ is odd and $t = \Node(t_L, (), t_R)$} \\
        \shape(t_L, k-1) \land t_R = \Leaf & \text{if $k$ is even and $t = \Node(t_L, (), t_R)$} \\
     \end{cases}
\end{equation*}

The $R_{\alpha}$ function is the (computable) recognizer of $\langle
h, t \rangle$ pairs in the range of $\alpha$. It is obvious that
$S_{p}$ is finite, and that either $t \in S_{p}$ or
$M_{p}(\alpha(t))$.

Next, we must show that if $M_{p}(\langle h, t \rangle)$ then
$|\alpha^{-1}(\langle h, t \rangle)| > p$. We note that if $h$ is
even, then $t$ has a right-hand subtree that is $\Leaf$, and there are
$\stcount(h-1)$ such subtrees that can be mapped to $\Leaf$ via the
catamorphism. Similarly, if $h$ is odd, then one of $t$'s children
will have a right-hand subtree that is $\Leaf$ (note that $h > 5$, thus there exists such a subtree), so there are at least
$\stcount(h-2)$ such subtrees. Finally, we note that for all values $k
> 5$, $\stcount(k-1) > \stcount(k-2) > k$, so if $h > p + 5$,
$|\alpha^{-1}(\langle h, t \rangle)|~\geq \stcount(h-2) > h > p + 5 > p$. \qed
\end{proof}

\begin{theorem}
\label{thm:height_is_not_mono}
The $\almostid$ catamorphism is not monotonic.
\end{theorem}

\begin{proof} We will prove this theorem by contradiction.
Suppose $\almostid$ was monotonic for some height $h_{\alpha}$.  First, we note that for all trees $t\in \rdttype$ where element type is unit, $\beta(t)$ is finite, bounded by the finite number of trees of height $\height(t)$.

We choose an arbitrary odd height $h_0$ such that $h_0 \geq h_{\alpha}$. Let $t_{min,h_0}$ be the tree of height $h_0$ such that $\beta(t_{min,h_0}) = \minbeta(h_0)$, which means:
\[\forall t \in \tau, \height(t) = h_0: \beta(t) \geq \beta(t_{min,h_0})\]
We can extend the tree $t_{min,h_0}$ to a new tree $t_{bad,h_0+1} = \Node(t_{min,h_0}, (), \Leaf)$, which has (even) height $h_0+1$. 
We can construct a bijection from every tree $t' \in \alpha^{-1}\bigl(\alpha(t_{bad,h_0+1})\bigr)$ to a tree in $\alpha^{-1}\bigl(\alpha(t_{min,h_0})\bigr)$ by extracting the left subtree of $t'$ (by construction, every right subtree of a tree in $\alpha^{-1}\bigl(\alpha(t_{bad,h_0+1})\bigr)$ is $\Leaf$).  Therefore, $\beta(t_{bad,h_0+1}) = \beta(t_{min,h_0})$.
Due to the construction of $t_{min,h_0}$, we have:
\[\forall t \in \rdttype, \height(t) = h_0: \beta(t) \geq \beta(t_{bad,h_0+1})\]
which implies that we cannot find any tree of height $h_0$ to satisfy the condition for monotonic catamorphisms in Definition \ref{definition:mono_cata} for tree $t_{bad, h_0+1}$ of height $h_0 + 1$.
\qed
\end{proof}

%



\subsection{A Note on Combining Monotonic Catamorphisms}
\label{section:combine_mono_catas}
One might ask if it is possible to have multiple monotonic
catamorphisms in the input formula while still maintaining the
completeness of the decision procedure. In general, when we combine
multiple monotonic catamorphisms, the resulting catamorphism might not
be monotonic or even GSS; therefore, the completeness of the decision
procedure is not guaranteed. For example, consider the monotonic
catamorphisms $\emph{List}_\emph{preorder}$ and $\emph{Sortedness}$ (their
definitions are in Table \ref{table:POPL_catas}). For any $h \in
\naturalnums^+$ we can construct a right skewed tree tree $t$ of
height $h$ as follows:
\[\Node\Bigl(\Leaf, 1, \Node\bigl(\Leaf, 2, \Node(\Leaf, 3, \ldots \Node(\Leaf, h , \Leaf))\bigr)\Bigr)\]
\noindent The values of $\emph{List}_\emph{preorder}(t)$ and $\emph{Sortedness}(t)$ are as follows:
\begin{itemize}
 \item $\emph{List}_\emph{preorder}(t) = (1~2~\ldots~h)$ -- i.e., the element values are $1, 2, \ldots, h$.
 \item $\emph{Sortedness}(t) = (1, h, \kw{true})$ -- i.e., $\min=1, \max=h$, and $t$ is sorted.
\end{itemize}
Let $\alpha$ be the combination of these catamorphisms and let $\beta$ be defined as in Definition \ref{definition:beta}.
We have $\beta(t) = 1$ as $t$ is the only tree that can map to the
values of $\emph{List}_\emph{preorder}(t)$ and $\emph{Sortedness}(t)$ above.
Thus, $\alpha$ cannot be monotonic or even GSS.

Although the combinability is not a feature that monotonic
catamorphisms can guarantee, Section \ref{section:assoc_cata} presents
a subclass of monotonic catamorphisms, called associative
catamorphisms, that supports the combination of catamorphisms in our
procedure.



\section{Associative Catamorphisms}
\label{section:assoc_cata}
We have presented an unrolling-based decision procedure that is
guaranteed to be both sound and complete with GSS catamorphisms (and
therefore also with sufficiently surjective and monotonic
catamorphisms). When it comes to catamorphisms, there are many
interesting open problems, for example: when is it possible to combine
catamorphisms in a complete way, or how computationally expensive is
it to solve catamorphism problems? This section attempts to
characterize a useful class of ``combinable" GSS catamorphisms that
maintain completeness under composition.

We name this class \emph{associative} catamorphisms due to the
associative properties of the operator used in the catamorphisms.
Associative catamorphisms have some very powerful important
properties: they are detectable\footnote{{\em detectable} in this
  context means that it is possible to determine whether or not a
  catamorphism is an associative catamorphism using an SMT solver.},
combinable, and impose an exponentially small upper bound on the
number of unrollings. Many catamorphisms presented so far are in fact
associative.

\begin{definition}[Associative catamorphism]
\label{definition:assoc_cata}
A catamorphism $\alpha:\rdttype \rightarrow \collection$ is
associative if
\begin{equation*}
\alpha(\Node(t_L, e, t_R)) = \alpha(t_L) \operator \etoc(e) \operator \alpha(t_R)
\end{equation*}
where $\operator:(\collection, \collection)\rightarrow \collection$ is
an associative binary operator. Here, $\etoc: \elem \rightarrow
\collection$ is a function that maps\footnote{For instance, if $\elem$
  is \kw{Int} and $\collection$ is \kw{IntSet}, we can have $\etoc(e)
  = \{e\}$.} an element value in $\elem$ to a corresponding value in
$\collection$.
\end{definition}


Associative catamorphisms are detectable. A catamorphism, written in
the format in Definition \ref{definition:assoc_cata}, is associative
if the $\operator$ operator is associative. This condition can be
easily proved by SMT solvers \cite{Barrett2011CVC4,DeMoura2008ZES} or
theorem provers such as ACL2 \cite{kaufmann2000computer}. Also,
because of the associative operator $\operator$, the value of an
associative catamorphism for a tree is independent of the shape of the
tree.

We present associative catamorphisms syntactically in
Definition~\ref{definition:assoc_cata}. They can also be described
semantically by requiring that $\alpha$ is preserved by tree rotations:
\begin{equation*}
\alpha\bigl(\Node(t_1, e_1, \Node(t_2, e_2, t_3))\bigr) =
\alpha\bigl(\Node(\Node(t_1, e_1, t_2), e_2, t_3))\bigr)
\end{equation*}
This is still detectable by checking the satisfiability the
corresponding query:
\begin{multline*}
R_\alpha(c_1) \land R_\alpha(c_2) \land R_\alpha(c_3) \land~\\
\kw{combine}(c_1, e_1, \kw{combine}(c_2, e_2, c_3)) \neq
\kw{combine}(\kw{combine}(c_1, e_1, c_2), e_2, c_3)
\end{multline*}
This semantic definition of associativity is broader than the purely
syntactic one (because it does not depend on the associative binary operator $\operator$), but is less intuitive. We work with the syntactic
definition in this section, but the main results over to the semantic
definition as well.

\begin{corollary}[Values of associative catamorphisms]
\label{corollary:assoc_cata_depends_elements}
The value of $\alpha(t)$, where $\alpha$ is an associative
catamorphism, only depends on the ordering and values of elements in
$t$. In particular, $\alpha(t)$ does not depend on the shape of the
tree:
\begin{equation*}
\alpha(t) = \alpha(\Leaf) \operator \etoc(e_1) \operator \alpha(\Leaf)
\operator \etoc(e_2) \operator \cdots \operator \alpha(\Leaf)
\operator \alpha(e_n) \operator \alpha(\Leaf)
\end{equation*}
where $e_1, e_2, \ldots, e_n$ is the in-order listing of the elements
of the nodes of $t$. When $t = \Leaf$, we simply have $\alpha(t) = \alpha(\Leaf)$.
\end{corollary}
\begin{proof}
Straightforward induction on the structure of $t$. \qed
\end{proof}

\begin{example}[Associative catamorphisms]
\label{example:assoc_catas}
In Table \ref{table:POPL_catas}, \emph{Height}, \emph{Some},
$\emph{List}_\emph{preorder}$ and $\emph{List}_\emph{postorder}$ are not
associative because their values depend on the shape of the tree.

The other catamorphisms in Table \ref{table:POPL_catas} are
associative, including \emph{Set}, \emph{Multiset}, \emph{SizeI},
$\emph{List}_\emph{inorder}$, \emph{Min}, and \emph{Sortedness} (both
with and without duplicates). The $\DW$ catamorphism in Section
\ref{section:dp_by_example} is also associative, where the operator
$\operator$ is $+$ and the mapping function is $\etoc(e) =
\bigl({\ite}~\dirty(e)~1~0\bigr)$. For $\emph{Multiset}$, the two
components are $\uplus$ and $\etoc(e) = \{e\}$.

Furthermore, we can define associative catamorphisms based on
associative operators such as $+, \cap, \max, \vee, \wedge$, etc. We
can also use a user-defined function as the operator in an associative
catamorphism. For example, the catamorphism \emph{Leftmost} which
finds the leftmost element value in a tree is associative where
$\etoc(e) = \Some(e)$, $\alpha(\Leaf) = \None$, and $\operator$ is
defined by
\begin{align*}
\None \operator \None &= \None \\
\Some(e) \operator \None &= \Some(e) \\
\None \operator \Some(e) &= \Some(e) \\
\Some(e_L) \operator \Some(e_R) &= \Some(e_L)
\end{align*}
The symmetrically defined \emph{Rightmost} catamorphism is also associative.

We do not require that $\alpha(\Leaf)$ is an identity for the operator
$\operator$, though it often is in practice. An example where it is
not is the \emph{Size} catamorphism which computes the total size of a
tree (rather than just the number of internal nodes computed by
\emph{SizeI}). In this case we have $\etoc(e) = 1$, $\alpha(\Leaf) =
1$, and operator $\operator$ is $+$. \exampleEndMark
\end{example}

\subsection{The Monotonicity of Associative Catamorphisms}
\label{section:assoc_cata_are_monotonic}

This section shows that associative catamorphisms are monotonic, and
therefore sufficiently surjective and GSS.

\begin{theorem}
\label{theorem:assoc_cata_are_monotonic}
Associative catamorphisms are monotonic.
\end{theorem}
\begin{proof}
Let $\halpha = 2$. Let $t$ be a tree with $\height(t) \geq 2$. If
$\beta(t) = \infty$ then we are done. Otherwise, suppose $\beta(t) <
\infty$. We want to show that there exists a tree $t_0$ such that
$\height(t_0) = \height(t) - 1$ and $\beta(t_0) < \beta(t)$. Since
$\height(t) \geq 2$ we can write $t = \Node(t_L, e, t_R)$ where either
$\height(t_L) = \height(t) - 1$ or $\height(t_R) = \height(t) - 1$.
Without loss of generality assume $\height(t_R) = \height(t) - 1$; the
argument is symmetric for the other case. We will show
that $\beta(t_R) < \beta(t)$ so that $t_R$ satisfies the conditions
required for $t_0$.

There are $\beta(t_R)$ trees that map to $\alpha(t_R)$. For each such
tree $t_R'$ the tree $t' = \Node(t_L, e, t_R')$ maps to $\alpha(t)$.
The distinctness of each $t_R'$ ensures that each $t'$ is also
distinct. Now all we need to find is one additional tree which maps to
$\alpha(t)$ but is not one of the $t'$ above. Since $\height(t) \geq
2$, we know $\height(t_R) \geq 1$ and can write $t_R = \Node(t_{RL},
e_R, t_{RR})$. Consider the rotated tree $t_{new} = \Node(\Node(t_L,
e, t_{RL}), e_R, t_{RR})$. This tree is distinct from the $t'$ trees
above since the left branches are distinct: $\Node(t_L, e, t_{LR})
\neq t_L$. Moreover, since $\operator$ is associative we have
\begin{align*}
\alpha(t_{new}) &= \bigl(\alpha(t_L) \operator \etoc(e) \operator
\alpha(t_{RL})\bigr) \operator \etoc(e_R) \operator \alpha(t_{RR}) \\
&= \alpha(t_L) \operator \etoc(e) \operator
\bigl(\alpha(t_{RL}) \operator \etoc(e_R) \operator \alpha(t_{RR})\bigr) \\
&= \alpha(t_L) \operator \etoc(e) \operator \alpha(t_R) \\
&= \alpha(t)
\end{align*}
Thus $\beta(t_R) < \beta(t)$, and therefore $\alpha$ is monotonic. \qed
\end{proof}

Since associative catamorphisms are monotonic, they are also GSS by
Theorem~\ref{theorem:monotonic_generalized_sufficient_surjectivity},
meaning that associative catamorphisms can be used in our unrolling
decision procedure.

\begin{remark}
Thus far we have used binary trees as our inductive datatype
$\rdttype$. Our results so far have been generic for any inductive
datatypes, but Theorem~\ref{theorem:assoc_cata_are_monotonic} is not.
In particular, the theorem does not hold when $\rdttype$ is the list
datatype. Over the list datatype the catamorphism $Multiset$ is
associative, but not monotonic since $\beta_{Multiset}(\{0, 0, \ldots,
0\}) = 1$. The proof of Theorem~\ref{theorem:assoc_cata_are_monotonic}
fails since there is no rotate operation on lists as there is on
trees. Similarly, the unrolling bounds in the next section do not
necessarily hold for list-like datatypes.
\end{remark}

\subsection{Exponentially Small Upper Bound on the Number of Unrollings}
\label{section:assoc_exponential_small_num_unrollings}

In the proof of Theorem
\ref{theorem:monotonic_generalized_sufficient_surjectivity}, we showed
that monotonic catamorphisms admit a linear bound on the number of
unrollings needed to establish unsatisfiability in our procedure.
However, even for monotonic catamorphisms, the number of unrollings
may be large for a large input formula with many tree disequalities,
leading to a high complexity for the algorithm. This section shows
that for associative catamorphisms, the bound can be made
\emph{exponentially small}.
\begin{lemma}
\label{lemma:beta_numshapes}
If $\alpha$ is an associative catamorphism then $\forall t\in
\rdttype: \beta(t) \geq \numshapes\bigl(\size(t)\bigr)$
\end{lemma}

\begin{proof}
Consider any tree $t \in \rdttype$.
Let $L$ be a list of size $ni(t)$ such that $L_j$, where $1 \leq j \leq ni(t)$,
is equal to the value stored in the $j$-th internal node in $t$.

Property \ref{property:size_is_odd} implies that any shape of size $size(t)$ must have exactly
$ni(t)$ $\SNode$(s) and $nl(t)$ $\SLeaf$(s).
Let $sh_1, \ldots, sh_{\numshapes(size(t))}$ be all shapes of size $size(t)$.
From $sh_i$, where $1 \leq i \leq \numshapes\bigl(size(t)\bigr)$, we construct a tree $t_i$ by converting every $\SLeaf$ in $sh_i$ into a $\Leaf$
and converting the $j$-th $\SNode$ in $sh_i$ into a structurally corresponding $\Node$ with element value $L_j$, where $1 \leq j \leq ni(t)$.
For example, the shape $\SNode\bigl(\SNode(\SLeaf, \SLeaf), \SLeaf\bigr)$ will be converted into the tree $\Node\bigl(\Node(\Leaf, L_2, \Leaf), L_1, \Leaf\bigr)$.
%

After this process, $t_1, \ldots, t_{\numshapes(size(t))}$ are mutually different because their shapes $sh_1$, $\ldots$, $sh_{\numshapes(size(t))}$ are distinct.
From Corollary \ref{corollary:assoc_cata_depends_elements}, we obtain
\[\alpha(t) = \alpha(t_1) = \ldots = \alpha(t_{\numshapes(size(t))}) = \etoc(L_1)~\operator~\etoc(L_2)~\operator~\ldots~\operator~\etoc(L_{ni(t)})\]
As a result, $\beta(t) \geq \numshapes\bigl(size(t)\bigr)$.
\qed
\end{proof}

\begin{lemma}
\label{lemma:minbeta_catalan}
If $\alpha$ is associative then $\forall h \in \naturalnums: \minbeta(h) \geq \catalan_h$.
\end{lemma}
\begin{proof}
Let $t_h \in \rdttype$ be any tree of height $h$. We have $\beta(t_h)
\geq \numshapes\bigl(\size(t_h)\bigr)$ from Lemma
\ref{lemma:beta_numshapes}. Hence, $\beta(t_h) \geq \numshapes(2h +
1)$ from Property \ref{property:size_height} and Lemma
\ref{lemma:numshapes}. From Lemma \ref{lemma:num_shapes_catalan},
$\beta(t_h) \geq \catalan_h$. Therefore, $\minbeta(h) \geq \catalan_h$
by Definition \ref{definition:minbeta}. \qed
\end{proof}

Let $h_p = \min\{h~|~\catalan_h > p\}$ so that $\catalan_{h_p} > p$.
From Lemma \ref{lemma:minbeta_catalan}, $\minbeta(h_p) \geq
\catalan_{h_p} > p$. Thus $h_p$ satisfies the GSS condition for
$\alpha$. Moreover, the growth of $\catalan_n$ is exponential
\cite{FlajoletSedgewick2009}.
Thus, $h_p$ is exponentially smaller than $p$ since $\catalan_{h_p} > p$.
For example, when $p = 10000$, we can choose $h_p = 10$ since $\catalan_{10} = 16796 > 10000$.
Similarly, when $p = 50000$, we can choose $h_p = 11$ since $\catalan_{11} = 58786$.

\subsection{Combining Associative Catamorphisms}
\label{section:combine_assoc_catas}

Let $\alpha_1, \ldots, \alpha_m$ be $m$ associative catamorphisms
where $\alpha_i$ is given by the collection domain $\collection_i$,
the operator $\operator_i$, and the function $\etoc_i$. We construct
the catamorphism $\alpha$ componentwise from the $\alpha_i$ as
follows:
\begin{itemize}
 \item $\collection$ is the domain of $m$-tuples, where the $i^\text{th}$
   element of each tuple is in $\collection_i$.
 \item $\operator:(\collection, \collection) \rightarrow \collection$
   is defined as
 \begin{equation*}
 \langle x_1, \ldots, x_m \rangle \operator \langle y_1, \ldots, y_m
 \rangle = \langle x_1 \operator_1 y_1, \ldots,
 x_m \operator_m y_{m}\rangle
  \end{equation*}
 \item $\etoc: \elem \rightarrow \collection$ is defined as $\etoc(e)
   = \bigl \langle \etoc_1(e), \ldots, \etoc_m(e) \bigr \rangle$
 \item $\alpha$ is defined as in Definition \ref{definition:assoc_cata}.
\end{itemize}

\begin{example}[Combine associative catamorphisms]
Consider \emph{Set} and \emph{SizeI} catamorphisms in Table
\ref{table:POPL_catas}, which are associative. When we combine the two
associative catamorphisms (assuming \emph{Set} is used before
\emph{SizeI}), we get a new catamorphism \emph{SetSizeI} that maps a
tree to a pair of values: the former is the set of all the elements in
the tree and the latter is the number of internal nodes in the tree.
For example, if we apply \emph{SetSizeI} to the tree in Fig.
\ref{fig:tree_example}, we get $\langle \{1, 2\}, 2 \rangle$.
\exampleEndMark
\end{example}

\begin{remark}
Every catamorphism obtained from the combination of associative
catamorphisms is also associative.
\end{remark}
\begin{proof}
The associativity of the componentwise $\operator$ follows directly
from the associativity of the $\operator_i$ operators. \qed
\end{proof}

Note that while it is easy to combine associative catamorphisms, it
might be challenging to compute the range predicate $R_{\alpha}$ of
the combination of those associative catamorphisms. For example,
consider $\emph{Min}$ and $\emph{Sum}$, two simple surjective
associative catamorphisms whose ranges are trivially equal to their
codomains. The range of their combination is:
\begin{align*}
\emph{Min}(t) &= \None~\wedge~\emph{Sum}(t) = 0\\
\vee~~\emph{Min}(t) &< 0\\
\vee~~\emph{Min}(t) &\geq 0~\wedge~\bigl(\emph{Sum}(t) = \emph{Min}(t) \vee \emph{Sum}(t) \geq 2 \times \emph{Min}(t)\bigr)
\end{align*}

\noindent As with individual catamorphisms, it is the user's responsibility to create an appropriate $R_{\alpha}$ predicate.




\section{The Relationship between Catamorphisms}
\label{section:catas_relationship}
We have summarized two types of catamorphisms previously proposed by
Suter et al. \cite{Suter2010DPA}, namely infinitely surjective and
sufficiently surjective catamorphisms in Definitions
\ref{definition:infinitely_surjective_catamorphisms} and
\ref{definition:sufficient_surjectivity}, respectively. We have also
proposed three different classes of catamorphisms: GSS
(Definition~\ref{definition:generalized_sufficient_surjectivity}),
monotonic (Definition \ref{definition:mono_cata}), and associative
(Definition \ref{definition:assoc_cata}). This section discusses how
these classes of catamorphisms are related to each other and how they
fit into the big picture, depicted in Fig.
\ref{fig:catas_relationship} with some catamorphism examples.

\begin{figure}[htb]
  \includegraphics[width=\textwidth]{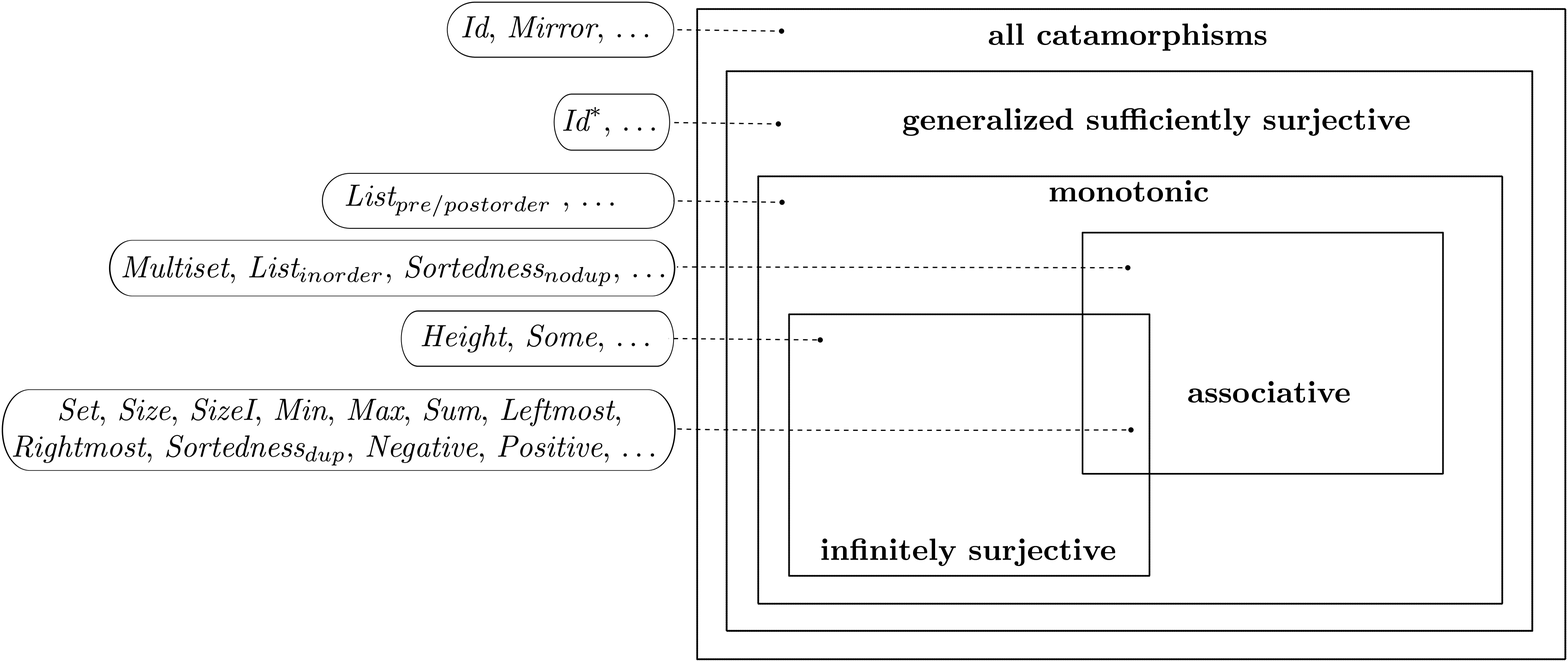}
  \caption{Relationship between different types of catamorphisms}
  \label{fig:catas_relationship}
\end{figure}

\emph{Between sufficiently surjective and GSS catamorphisms}: We have
shown that all sufficiently surjective catamorphisms are GSS
(Corollary~\ref{corollary:ss_is_gss}). We have not demonstrated that
the inclusion is strict as this would require reasoning about what is
and is not representable in the definition of $M_p$. Frankly, we
believe there is no need for sufficient surjectivity given the notion
of GSS.

\emph{Between monotonic and sufficiently surjective catamorphisms}:
All monotonic catamorphisms are sufficiently surjective (Theorem
\ref{theorem:mono_catas_are_suff_surj}). This shows that although the
definition of monotonic catamorphisms from this paper and the idea of
sufficiently surjective catamorphisms from Suter et al.
\cite{Suter2010DPA} may look different from each other, they are
actually closely related. Moreover, monotonic catamorphisms provide
linear bounds in our decision procedure.

\emph{Between infinitely surjective and monotonic catamorphisms}: All
infinitely surjective catamorphisms are monotonic (Lemma
\ref{lemma:infinitely_surjective_abstraction_monotonic}). Thus,
infinitely surjective catamorphisms are not just a sub-class of
sufficiently surjective catamorphisms as presented in Suter et al.
\cite{Suter2010DPA}, they are also a sub-class of monotonic
catamorphisms.

\emph{Between associative and monotonic catamorphisms}: All
associative catamorphisms are monotonic
(Theorem~\ref{theorem:assoc_cata_are_monotonic}). Moreover,
associative catamorphisms provide exponentially small bounds in our
decision procedure.

\emph{Between infinitely surjective and associative catamorphisms}:
The set of infinitely surjective catamorphisms and that of associative
catamorphisms are intersecting, as shown in Fig.
\ref{fig:catas_relationship} with some catamorphism examples.


\section{Implementation and Experimental Results}
\label{section:experimental_results}

\label{section:rada}

We introduce RADA\footnote{{\url{http://crisys.cs.umn.edu/rada/}}.},
an open source tool to reason about algebraic data types with
abstractions that is conformant with the \smtlib\
format \cite{BarSTSMT10}. The algorithms behind RADA were described in
previous sections. It can function as a back-end for reasoning about
recursive programs that manipulate algebraic data types. RADA was
designed to be host-language and solver-independent, and it can use
either CVC4 or Z3 as its underlying SMT solver. RADA has also been
successfully integrated into the Guardol system \cite{Hardin2012GLV},
replacing our implementation of the Suter-Dotta-Kuncak decision
procedure \cite{Suter2010DPA} on top of
OpenSMT \cite{BruttomessoPST10}. Experiments show that our tool is
reliable, fast, and works seamlessly across multiple platforms,
including Windows, Unix, and Mac OS. We have used RADA in the Guardol
project for reasoning about functional implementations of complex data
structures and to reason about {\em guard applications} that determine
whether XML messages should be allowed to cross network security
domains. How RADA was integrated into Guardol is presented
in \cite{Hardin2012GLV}.


The overall architecture of RADA follows closely the decision procedure described in Section \ref{section:revised_unrolling_procedure}. We use CVC4 \cite{Barrett2011CVC4} and Z3 \cite{DeMoura2008ZES} as the underlying SMT solvers in RADA due to their powerful abilities to reason about recursive data types. The grammar of RADA in Fig. \ref{fig:ourDPGrammar} is based on the \smtlib\ \cite{BarSTSMT10} format with some new syntax for selectors, testers, data type declarations, and catamorphism declarations.
Note that although selectors, testers, and data type declarations are not defined in \smtlib, all of them are currently supported by both CVC4 and Z3; thus, only catamorphism declarations are not understood by these solvers.
$\textbf{:post-cond}$, which is used to declare $R_{\alpha}$, is optional since we do not need to specify $R_{\alpha}$ when $\alpha$ is surjective (e.g., \textsf{SumTree} in Example \ref{example:rada_syntax}).


\newcommand{\bnfsymb}[1]{\langle #1 \rangle}
\begin{figure}[htb]
\small
\begin{tabular}{rcl}
\toprule
$\bnfsymb{command}_1$                         & ::= &  $\textbf{(}~\textbf{declare-datatypes}~\textbf{()}~\textbf{(}\bnfsymb{datatype}^+\textbf{)}~\textbf{)}$
\\
$\bnfsymb{datatype}$                          & ::= &  $\textbf{(}~\bnfsymb{symbol}~\bnfsymb{datatype\_branch}^+~\textbf{)}$
\\
$\bnfsymb{datatype\_branch}$                  & ::= &
$\textbf{(}~\bnfsymb{symbol}~\bnfsymb{datatype\_branch\_para}^*~\textbf{)}$
\\
$\bnfsymb{datatype\_branch\_para}$       & ::= &
$\textbf{(}~\bnfsymb{symbol}~\bnfsymb{sort}~\textbf{)}$
\\
\\
$\bnfsymb{command}_2$                         & ::= &
$\textbf{(}~\textbf{define-catamorphism}~\bnfsymb{catamorphism}~\textbf{)}$
\\
$\bnfsymb{catamorphism}$                      & ::= &
$\textbf{(}~\bnfsymb{symbol}~\textbf{(}~\bnfsymb{sort}~\textbf{)}~\bnfsymb{sort}~\bnfsymb{term}$
\\
& & \hfill $[\textbf{:post-cond}~\bnfsymb{term}]~\textbf{)}$\\
\\
$\bnfsymb{selector\_application}$              & ::= &
$\bnfsymb{symbol}~\bnfsymb{symbol}$
\\
$\bnfsymb{tester\_application}$                & ::= &
$\textbf{is-}\bnfsymb{symbol}~\bnfsymb{symbol}$
\\
\bottomrule
\end{tabular}
\caption{RADA grammar.}
\label{fig:ourDPGrammar}
\end{figure}
\begin{example}[RADA syntax]
 \label{example:rada_syntax}
Let us consider an example to illustrate the syntax used in RADA.
Suppose we have a data type \textsf{RealTree} that contains real numbers:

{\small
\begin{verbatim}
    (declare-datatypes ()
      ((RealTree
         (Leaf)
         (Node (left RealTree) (elem Real) (right RealTree)))))
\end{verbatim}
}
\noindent Next, a \textsf{RealTree} can be abstracted into a real number representing the sum of all elements in the tree by catamorphism \textsf{SumTree}, which can be defined as follows:
{ \small
\begin{verbatim}
    (define-catamorphism SumTree ((t RealTree)) Real
      (ite (is-Leaf t)
           0.0
           (+ (SumTree (left t))
              (elem t)
              (SumTree (right t)))))
\end{verbatim}
}
\noindent where \textsf{is-Leaf} is a tester that checks if a \textsf{RealTree} is a leaf node and \textsf{left t}, \textsf{elem t}, and \textsf{right t} are selectors that select the corresponding data type branches in a \textsf{RealTree} named \textsf{t}. Given the definitions of data type \textsf{RealTree} and catamorphism \textsf{SumTree}, one may want to check some properties of a \textsf{RealTree}, for example:

{\small
\begin{verbatim}
    (declare-fun t1 () RealTree)
    (declare-fun t2 () RealTree)
    (declare-fun t3 () RealTree)
    (assert (= t1 (Node t2 5.0 t3)))
    (assert (= (SumTree t1) 5.0))
    (check-sat)
\end{verbatim}
}
\noindent As expected, RADA returns \emph{SAT} for the above example.
\exampleEndMark
\end{example}
Since RADA was first published \cite{PhamRADA13}, we have been working on improving the performance of the tool.
Compared with the version in \cite{PhamRADA13}, the current version of
RADA is multiple times faster thanks to the following implementation techniques.



\emph{\underline{Technique 1}: Solve proof obligations in parallel.} Multiple proof obligations can be written in RADA within \kw{push}-\kw{pop} pairs (as in \smtlib\ \cite{BarSTSMT10}). For instance,

{\small
\begin{verbatim}
    (push) Obligation_A (pop)                  (push) Obligation_B (pop)
\end{verbatim}
}


We preprocess the original SMT file.
If the file has parallelizable obligations, we split it into multiple separate files (each file has only one obligation).
RADA discharges proof obligations in parallel. It supports a thread pool of a configurable size of proof obligations. All the proof obligations in the pool are solved concurrently and all the remaining proof obligations are put in a waiting list. As soon as a proof obligation in the thread pool is discharged, the pool adds a new proof obligation from the waiting list to the pool (if any).

\emph{\underline{Technique 2}: Reuse the definitions of catamorphism bodies when unrolling.} In general, when we have a catamorphism application, e.g., \kw{SumTree (Node t2 5.0 t3)} with the \kw{SumTree} catamorphism and tree terms \kw{t2} and \kw{t3} in Example \ref{example:rada_syntax}, the catamorphism application is assigned to the corresponding definition of the catamorphism body with the given parameter. In this case, it will be as follows:

{\small
\begin{verbatim}
    (assert (= (SumTree (Node t2 5.0 t3))
               (ite (is-Leaf (Node t2 5.0 t3))
                    0.0
                    (+ (SumTree (left (Node t2 5.0 t3)))
                       (elem (Node t2 5.0 t3))
                       (SumTree (right (Node t2 5.0 t3)))))))
\end{verbatim}
}

However, as the unrolling procedure progresses, the tree parameters will keep getting bigger (because they are unrolled) and the catamorphism applications will appear frequently in the SMT query. This leads to the following issue: the definitions of catamorphism bodies appear again and again.
To address this issue, it is desirable to be able to reuse the definitions of catamorphism bodies. To do that, RADA creates a user-defined function for each catamorphism body, for example with the \kw{SumTree} catamorphism:

{\small
\begin{verbatim}
    (define-fun SumTree_GeneratedCatDefineFun ((t RealTree)) Real
      (ite (is-Leaf t)
          0.0
          (+ (SumTree (left t))
             (elem t)
             (SumTree (right t)))))
\end{verbatim}
}
\noindent and whenever we want to calculate a catamorphism application, we just need to call the corresponding user-defined function we just created:

{\small
\begin{verbatim}
    (assert (= (SumTree (Node t2 5.0 t3))
               (SumTree_GeneratedCatDefineFun (Node t2 5.0 t3))))
\end{verbatim}
}
\noindent We can also parameterize the above equality assertion by creating another user-defined function for it as follows:

{\small
\begin{verbatim}
    (define-fun SumTree_GeneratedUnrollDefineFun ((t RealTree)) Bool
      (= (SumTree t) (SumTree_GeneratedCatDefineFun t)))
\end{verbatim}
}
\noindent and now all what we need to do is use the short, newly created function:

{\small
\begin{verbatim}
    (assert (SumTree_GeneratedUnrollDefineFun (Node t2 5.0 t3)))
\end{verbatim}
}

In other words, when we need to unroll a catamorphism application, we just need to call the corresponding function with suitable parameters instead of expanding tree terms repeatedly.

\emph{\underline{Technique 3}: Solve each proof obligation incrementally.} We observe that in our decision procedure, we need two calls to an SMT solver (i.e., two \emph{decide} calls in Algorithm \ref{alg:revised_dp}) at each unrolling step to determine if we have found a trustworthy \emph{SAT}/\emph{UNSAT} answer.
There are two issues if the calls to the SMT solver are handled independently: (1) we would not take advantage of what the SMT solver instance has learned from the previous SMT query, and (2) we would pay a performance price for initializing and closing the SMT solver instance each time.

RADA addresses those issues as follows.
First, RADA solves each proof obligation incrementally, i.e., the information collected from the SMT queries is reused over time. Second, there is only one instance of the SMT solver for each proof obligation we want to solve; in other words, RADA creates an instance of the SMT solver when we start solving the proof obligation and only closes the SMT solver instance after the obligation has been completely discharged. We show below an example of incremental solving with RADA.

\begin{example}[Example of incremental solving with RADA]
\label{appendix:incremental_solving_with_rada}
Let us present step by step how RADA solves the \kw{RealTree} example in Example \ref{example:rada_syntax}. First, RADA sends to an SMT solver the declaration of the \kw{RealTree} data type,
which is the \textsf{declare-datatypes} statement in Example \ref{example:rada_syntax}.


Next, RADA declares an uninterpreted function called \kw{SumTree}, which represents the \kw{SumTree} catamorphism in Example \ref{example:rada_syntax}.
Note that the SMT solver views \kw{SumTree} as an uninterpreted function: the solver does not know what content of the function is; it only knows that \kw{SumTree} takes as input a \kw{RealTree} and returns a \kw{Real} value as the output.

{\small
\begin{verbatim}
    (declare-fun SumTree (RealTree) Real)
\end{verbatim}
}
\noindent RADA then feeds to the SMT solver the original problem we want to solve:
{\small
\begin{verbatim}
    (declare-fun t1 () RealTree)
    (declare-fun t2 () RealTree)
    (declare-fun t3 () RealTree)
    (assert (= t1 (Node t2 5.0 t3)))
    (assert (= (SumTree t1) 5.0))
\end{verbatim}
}
\noindent Additionally, RADA creates two user-defined functions as previously discussed as a preprocessing step:

{\small
\begin{verbatim}
    (define-fun SumTree_GeneratedCatDefineFun ((t RealTree)) Real
      (ite (is-Leaf t)
          0.0
          (+ (SumTree (left t))
             (elem t)
             (SumTree (right t)))))

    (define-fun SumTree_GeneratedUnrollDefineFun ((t RealTree)) Bool
      (= (SumTree t) (SumTree_GeneratedCatDefineFun t)))
\end{verbatim}
}
\noindent RADA then tries to check the satisfiability of the problem without unrolling any catamorphism application:
{\small
\begin{verbatim}
    (check-sat)
\end{verbatim}
}
\noindent The SMT solver will return \emph{SAT}. In this case, we are using the uninterpreted function; hence, the \emph{SAT} result is untrustworthy.
Therefore, we have to continue the process by unrolling the catamorphism application \kw{SumTree t1}. We also add a \kw{push} statement and then add the control conditions to the problems before checking its satisfiability. Note that the \kw{push} statement is used here to mark the position in which the control conditions are located, so that we can remove the control conditions later by a corresponding \kw{pop} statement.

{\small
\begin{verbatim}
    (assert (SumTree_GeneratedUnrollDefineFun t1))                [Unrolling step]
    (push)
    (assert (is-Leaf t1))                      [Assertions for control conditions]
    (check-sat)
\end{verbatim}
}
\noindent The SMT solver will return \emph{UNSAT}, which means using the control conditions might be too restrictive and we have to remove the control conditions by using a \kw{pop} statement and try again:

{\small
\begin{verbatim}    
    (pop)                                          [Remove the control conditions]
    (check-sat)
\end{verbatim}
}
\noindent However, when checking the satisfiability without control conditions, we get \emph{SAT} from the SMT solver again. Based on our decision procedure in Algorithm \ref{alg:revised_dp}, we have to try another unrolling step; thus, RADA sends the following to the solver:

{\small
\begin{verbatim}
    (assert (SumTree_GeneratedUnrollDefineFun (left t1)))         [Unrolling step]
    (assert (SumTree_GeneratedUnrollDefineFun (right t1)))
    (push)
    (assert (is-Leaf (left t1)))               [Assertions for control conditions]
    (assert (is-Leaf (right t1)))
    (check-sat)
\end{verbatim}
}
\noindent This time the SMT solver still returns \emph{SAT}. However, we are using control conditions and getting \emph{SAT}, which means the \emph{SAT} result is trustworthy. Thus, RADA returns \emph{SAT} as the answer to the original problem. This example has shown how we can use only one SMT solver instance to solve the problem incrementally.
\exampleEndMark
\end{example}





\subsection{Experimental Results}
\label{section:rada_experimental_results}
We have implemented our decision procedure in RADA and evaluated the
tool with a collection of benchmark guard examples listed in Table
\ref{table:experimental_results}. All of the benchmark examples were
automatically verified by RADA in a short amount of time.

\begin{table}[htb]
\caption{Experimental results}
\scriptsize
\begin{tabular}{clcc}
  \toprule
   Type & Benchmark & Result & Time (s)\\
   \midrule
   Single & sumtree$(01|02|03|05|06|07|10|11|13)$ & \textsf{sat} & 0.025--0.083 \\
   associative & sumtree$(04|08|09|12|14)$ & \textsf{unsat} &
   0.033--0.044 \\
   catamorphisms \\
   \midrule
   & min\_max$(01|02)$ & \textsf{unsat} & 0.057--0.738 \\
   Combination of & min\_max\_sum01 & \textsf{unsat} & 1.165 \\
   associative & min\_max\_sum$(02|03|04)$ & \textsf{sat} & 0.149 -- 0.373\\
   catamorphisms & min\_size\_sum01 & \textsf{unsat} & 0.873 \\
   & min\_size\_sum02 & \textsf{sat} & 0.114 \\
   & negative\_positive$(01|02)$ & \textsf{unsat} & 0.038 -- 0.136\\
   \midrule
   & Email\_Guard\_Correct\_All & 17 \textsf{unsats} & $\approx$ 0.009/obligation \\
   & RBTree.Black\_Property & 12 \textsf{unsats} & $\approx$ 2.142/obligation \\
   Guardol & RBTree.Red\_Property & 12 \textsf{unsats} & $\approx$ 0.163/obligation\\
   & array\_checksum.SumListAdd & 2 \textsf{unsats} &  $\approx$ 0.028/obligation\\
   & array\_checksum.SumListAdd\_Alt & 13 \textsf{unsats} & $\approx$ 0.012/obligation \\
   \bottomrule
\end{tabular}
\label{table:experimental_results}
\end{table}

\emph{{Experiments on associative catamorphisms}.} The first set of
benchmarks consists of examples related to \emph{Sum}, an associative
catamorphism that computes the sum of all element values in a tree.
The second set contains combinations of associative catamorphisms that
are used to verify some interesting properties such as (1) there does
not exist a tree with at least one element value that is both positive
and negative and (2) the minimum value in a tree cannot be bigger than
the maximum value in the tree. The definitions of the associative
catamorphisms used in the benchmarks are as follows: \emph{Sum} is
defined as in Example \ref{example:rada_syntax}, \emph{Max} is defined
in a similar way to \emph{Min} in Table \ref{table:POPL_catas}, and
\emph{Negative} and \emph{Positive} are defined as in
\cite{Pham2013PAC}.

\emph{{Experiments on Guardol benchmarks}.} In addition to associative
catamorphisms, we have also evaluated RADA on some examples in the
last set of benchmark containing general catamorphisms that have been
automatically generated from the Guardol verification system
\cite{Hardin2012GLV}. They consist of verification conditions to prove
some interesting properties of red black trees and the checksums of
trees of arrays. These examples are complex: each of them contains
multiple verification conditions, some data types, and a number of
mutually related parameterized catamorphisms. For example, the Email
Guard benchmark has 8 mutually recursive data types, 6 catamorphisms,
and 17 complex obligations.

All benchmarks were run on a Ubuntu machine (Intel Core I5, 2.8 GHz, 4GB RAM).
All running time was measured when Z3 was used as the underlying SMT solver.


\section{Conclusion and Discussion}
\label{section:conclusion}

In this paper, we have proposed an unrolling-based decision procedure
for algebraic data types with a new idea of generalized sufficiently
surjective catamorphisms. We have also presented a class of
generalized sufficiently surjective catamorphisms called monotonic
catamorphisms and have shown that all sufficiently surjective
catamorphisms known in the literature to date \cite{Suter2010DPA} are
also monotonic. We have established a linear upper bound on the number
of unrollings needed to establish unsatisfiability with monotonic
catamorphisms. Furthermore, we have pointed out a sub-class of
monotonic catamorphisms, namely associative catamorphisms, which are
proved to be detectable, combinable, and guarantee an exponentially
small unrolling bound thanks to their close relationship
with Catalan numbers. Our combination results extend the set of
problems that can easily be reasoned about using the
catamorphism-based approach.

We have also presented RADA, an open source tool to reason about
inductive data types. RADA fully supports all types of catamorphisms
discussed in this paper as well as other general user-defined
abstraction functions. The tool was designed to be simple, efficient,
portable, and easy to use. The successful uses of RADA in the Guardol
project \cite{Hardin2012GLV} demonstrate that RADA not only could
serve as a good research prototype tool but also holds great promise
for being used in other real world applications.

\section*{Compliance with Ethical Standards}
\begin{itemize}
\item Conflict of Interest: We declare that we have no conflict of interest.
\item Research involving Human Participants and/or Animals: We declare that this research does not involve human participants and/or animals.
\item Informed Consent: We declare that no informed consent is needed since this research does not involve human participants.
\end{itemize}


\end{document}